\documentclass[11pt]{article}
\usepackage{graphicx} 
\usepackage{style}

\title{A Nearly Optimal Min-Cut Algorithm for Dynamic Streams}
\author{
Yang P. Liu\footnote{(yangl7@andrew.cmu.edu) Carnegie Mellon University. Supported in part by NSF Career Award CCF2542272.\smallskip} 
\and 
Hoai-An Nguyen\footnote{(hnnguyen@andrew.cmu.edu) Carnegie Mellon University. Supported in part by
an NSF GRFP fellowship grant number DGE2140739/DGE2631988, NSF Career Award CCF2542272, Office
of Naval Research award number N000142112647, and a Simons Investigator Award.
\smallskip} 
\and 
Vihan Shah\footnote{(vihanshah98@gmail.com) University of Birmingham. Supported in part by the European Union under the Horizon Europe Marie Skłodowska-Curie Actions, Grant Agreement No. 101282993 (HuNTR).\smallskip}
\and 
David P. Woodruff\footnote{(dwoodruf@cs.cmu.edu) Carnegie Mellon University, Google Research. Supported in part by Office
of Naval Research award number N000142112647 and a Simons Investigator Award.}
}
\date{}

\begin{document}

\maketitle

\thispagestyle{empty}
\pagenumbering{roman}

\begin{abstract}
    We give an $\Ot(n/\eps)$-space algorithm for the minimum cut problem in dynamic streams, resolving the space complexity of streaming min-cut up to polylogarithmic factors. Ding, Garces, Li, Lin, Nelson, Shah, and Woodruff (ITCS 2025) obtained the same $\Ot(n/\eps)$ space bound for insertion-only streams, while prior to our work the best known bound even for unweighted dynamic streams was $\Ot(n/\eps^2)$, due to Ahn, Guha, and McGregor (PODS 2012). Our main technical contribution is an $\Ot(n/\eps)$-space for-each cut sketch in dynamic streams, which may be of independent interest. The construction applies to weighted graphs under general $\pm\Delta$ updates, a stronger model than the standard dynamic-stream setting in which weighted edges are inserted and deleted in their entirety. Despite this stronger guarantee, our algorithm is surprisingly simple and relies on very little machinery. We also obtain applications to minimum $\ell_p$-norm $k$-partitioning for every fixed $k$, weighted correlation-clustering evaluation for two clusters, and for-each cut sketching in balanced directed graphs.
\end{abstract}

\clearpage

\pagenumbering{arabic}
\setcounter{page}{1}

\section{Introduction}
Minimum cut is one of the most fundamental graph problems and has been studied extensively across every major model of computation. Beyond its theoretical interest, minimum cut also has practical applications, including measuring influence in social networks~\cite{wang2020efficient} and robustness to failures in power~\cite{hendrickxpowernetwork} and road networks~\cite{sugiuraroadnetwork}.

In the classical setting, a line of work~\cite{karger1993global,ks96,bk96}, culminating in~\cite{karger00}, gave the first randomized algorithm for minimum cut running in $O(m \log^3 n)$ time, where $m$ is the number of edges in the graph. 
Many years later, Gawrychowski, Mozes, and Weimann~\cite{gawrychowski2020minimum} improved the randomized running time to $O(m\log^2 n)$.
Around the same time, Li and Panigrahi~\cite{li2020deterministic} showed that minimum cut can be solved deterministically within a polylogarithmic factor of max-flow time. Breakthroughs in max flow~\cite{chen2022maximum,van2023deterministic} then implied an $m^{1+o(1)}$-time deterministic algorithm for minimum cut. Finally, a line of work~\cite{kawarabayashi2018deterministic,henzinger2020local,li2021deterministic,henzinger2024deterministic} showed that minimum cut can be solved deterministically in $\Ot(m)$\footnote{Throughout the paper, we use $\Ot(f) := O(f \cdot \poly \log(n))$ to hide polylogarithmic factors in $n$.} time.

We study the minimum cut problem in the graph streaming model, which is particularly relevant for processing massive graphs. We consider the arbitrary-order streaming setting, where the vertex set is fixed and edge insertions and deletions arrive sequentially in a stream in an arbitrary order. The goal is to solve the problem using as little space as possible. The streaming model was first formalized by Alon, Matias, and Szegedy~\cite{alonms96} and has been studied extensively over the past couple of decades (see, e.g.,~\cite{mcgregor2014graph} for a survey on graph streams).

We first review the setting of simple unweighted graphs. Zelke~\cite{zelke2011intractability} showed that solving minimum cut exactly in the streaming model requires $\Omega(n^2)$ bits of space. Since the entire graph can be stored using $\binom{n}{2}$ bits, beating the quadratic bound requires settling for an approximation. Ahn and Guha~\cite{ahn2009graph} showed that a $(1+\eps)$-approximation to minimum cut can be achieved using $\Ot(n/\eps^2)$ space, and also proved an $\Omega(n\log(1/\eps))$-bit lower bound for this problem. These bounds hold in the insertion-only streaming model where the stream only contains edge insertions.
Ahn, Guha, and McGregor~\cite{ahngm12graph} subsequently extended the $\Ot(n/\eps^2)$-space upper bound to insertion-deletion streams for unweighted graphs. They also obtained a weighted extension by grouping edges into weight classes, but this requires the weight of an edge to be known when that edge is inserted or deleted and therefore does not apply to general turnstile weight updates. Chen, Khanna, and Li~\cite{CKL2022weighted} later gave the first linear-sketching algorithm for weighted graphs in the general turnstile model, using $\Ot(n/\eps^3)$ space.

These approaches are based on the notion of cut sparsifiers introduced by Benczúr and Karger \cite{bk96}. A cut sparsifier of a graph is a sparse weighted subgraph that approximately preserves the value of every cut. For a \emph{for-all} guarantee, a single random sketch must simultaneously approximate every cut with high probability. In contrast, a \emph{for-each} guarantee only requires that, for any fixed cut chosen independently of the sketch randomness, that cut is approximated with high probability.

These prior streaming algorithms achieve their guarantees by constructing for-all cut sparsifiers using $\Ot(n/\eps^2)$ or $\Ot(n/\eps^3)$ space. 
Benczúr and Karger~\cite{bk96} showed that every graph admits a cut sparsifier with $O(n\log n/\eps^2)$ edges, and this was later improved to $O(n/\eps^2)$ edges by~\cite{bss12}. On the other hand, an $\Omega(n\log n/\eps^2)$-bit lower bound for representing a cut sparsifier was established in a line of work~\cite{alon1997edge,andonickqwz16,carlson2019optimal}, even when the sparsifier is not required to be a subgraph of the original graph.  The corresponding existence results were also extended to the stronger notion of spectral sparsifiers~\cite{st04,st11,spi11,bss12}. Consequently, any streaming algorithm for minimum cut that explicitly constructs such a cut sparsifier incurs $\Omega(n\log n/\eps^2)$ bits of space.

This limitation of the cut-sparsifier approach was later circumvented by Ding, Garces, Li, Lin, Nelson, Shah, and Woodruff \cite{ding2025}, who gave a minimum cut algorithm for weighted graphs (where the weight of an edge is known upon insertion) using $\Ot(n/\eps)$ bits of space in insertion-only streams. A key ingredient in their algorithm is a streaming construction of a for-each spectral sparsifier, building on a line of work~\cite{andonickqwz16,jambulapati2018efficient,chu2020graph}. 

More specifically, the for-each spectral sparsifier of~\cite{ding2025} builds on the short-cycle decomposition framework of~\cite{chu2020graph}. At a high level, the approach finds a large collection of edge-disjoint short cycles, randomly keeps and reweights one alternating half of each cycle, and repeats on the remaining graph. Constructing the required short-cycle decompositions is itself highly nontrivial, relying on low-congestion cycle covers and additional graph-decomposition machinery. Thus, even the core sparsification primitive behind their result requires a substantial collection of sophisticated spectral and graph-theoretic tools.

\cite{ding2025} also proved an $\Omega(n/\eps)$-bit lower bound for minimum cut, resolving the space complexity of the problem in insertion-only streams up to logarithmic factors. However, even for simple unweighted graphs in dynamic streams, the best known space bound remained $\Ot(n/\eps^2)$ bits.

This gap in our understanding is particularly interesting in light of a broader phenomenon in graph streaming, where many problems have essentially the same space complexity in insertion-only and dynamic streams up to logarithmic factors. This insertion-only versus dynamic ``near-equivalence'' is known for several problems, including $k$-edge-connectivity~\cite{nagamochi1992linear,ahngm12}, $k$-vertex-connectivity~\cite{cheriyan1993scan,assadi2023tight}, cut sparsifiers~\cite{ahn2009graph,ahngm12graph}, spectral sparsifiers~\cite{kelner2013spectral,kapralov2017single}, constant-factor approximation to densest subgraph~\cite{mcgregortvv15,esfandiari2015applications}, correlation clustering~\cite{assadi2025correlation,assadi2023streaming}, and $(\Delta+1)$-vertex coloring~\cite{assadick19a}. Notable exceptions include maximum matching and minimum vertex cover, for which dynamic streams are provably harder than insertion-only streams~\cite{assadikly16,darkk20,assadi2022asymptotically,naidus22space}.
Approximate minimum cut therefore represents one of the remaining gaps in our understanding of the relationship between insertion-only and dynamic streams:
\begin{question}
\label{q:dynamic-min-cut}
Can $(1+\eps)$-approximate minimum cut be solved in $\Ot(n/\eps)$ bits of space in dynamic graph streams where both insertions and deletions are allowed?
\end{question}

\subsection{Our results}

We answer \Cref{q:dynamic-min-cut} in the affirmative. In fact, our algorithm works in the more general weighted turnstile model, where each update $(e,\Delta)$ changes the weight of edge $e$ by an arbitrary integer $\Delta$.

\begin{result} [Formalized in \Cref{thm:formal-weighted}]
There is a one-pass streaming algorithm that, given a weighted graph under general turnstile updates, outputs a $(1+\eps)$-approximate minimum cut with high probability using $\Ot(n/\eps)$ bits of space. Take the update time of a constant-approximation for-all cut sparsifier to be $t$. Then the update time is $t + \poly \log n$. 
\end{result}

Our main technical contribution is an $\Ot(n/\eps)$-space \emph{for-each cut sketch}\footnote{
We use the term \emph{for-each cut sketch} rather than \emph{for-each cut sparsifier} because our data structure need not output a weighted subgraph of the input graph. Instead, for any fixed cut chosen independently of the sketch randomness, it returns an approximation to the value of that cut.} for weighted graphs under general turnstile updates. We note that our sketch also achieves $\poly \log n$ update time. The overall result uses a \emph{for-all} cut sparsifier as a blackbox, and therefore the update time of the whole algorithm is $t + \poly \log n$ where $t$ is the update time of the for-all sparsifier\footnote{For unweighted graphs, or weighted graphs under the more restricted kind of updates where an edge is inserted or deleted with its weight, the overall update time is $\poly \log n$ due to the existence of a for-all cut sparsifier in this setting with update time $\poly \log n$.}. In the offline setting, where the graph is given explicitly, Andoni, Chen, Krauthgamer, Qin, Woodruff, and Zhang~\cite{andonickqwz16} previously gave an $\Ot(n/\eps)$-size for-each cut sketch. We obtain the same space bound with a sketch that can be maintained in a single pass over a dynamic stream. Moreover, our update model is strictly stronger than the weighted insertion--deletion setting typically considered in dynamic graph streams, where each edge is inserted or deleted together with its full weight.

The general weighted turnstile model is substantially more challenging: even for the extensively studied problem of constructing a \emph{for-all} cut sparsifier, the optimal space complexity is not known. The best known upper bound is $\Ot(n/\eps^3)$, due to Chen, Khanna, and Li~\cite{CKL2022weighted}, compared with the $\widetilde{\Omega}(n/\eps^2)$ lower bound inherent to cut sparsification. Thus, there remains an extra factor of $1/\eps$ even for this fundamental and heavily studied primitive. For both minimum cut and for-each cut sketching, our $\Ot(n/\eps)$ space bound is optimal up to polylogarithmic factors: \cite{ding2025} prove an $\Omega(n/\eps)$ lower bound for minimum cut, while \cite{andonickqwz16} prove the corresponding lower bound for for-each cut sketches.

A notable feature of our construction is its \emph{simplicity}. Prior near-optimal for-each cut and spectral sketches rely on recursive decompositions into high-expansion pieces~\cite{andonickqwz16,jambulapati2018efficient}, while the recent insertion-only minimum-cut algorithm of~\cite{ding2025} builds on the short-cycle decomposition framework of~\cite{chu2020graph}. Our construction requires neither type of graph decomposition. Instead, it uses only a small collection of standard streaming primitives: exact counters, $\ell_1$-samplers, and constant-dimensional Cauchy sketches for estimating the $\ell_1$ norm. 
Moreover, in the insertion-only streaming setting, the underlying $\ell_1$-sampling and $\ell_1$-estimation primitives admit simpler implementations, making our framework an especially elementary alternative to the prior near-optimal insertion-only algorithm of \cite{ding2025}. Thus, despite operating in the stronger turnstile model, the resulting sketch has a comparatively simple structure.

We emphasize that our construction is a \emph{cut} sketch rather than a spectral sketch. Our techniques exploit the structure of a single queried cut and do not extend directly to preserving spectral information. This distinction is important in the weighted turnstile setting. Chen, Khanna, and Li~\cite{CKL2022weighted} proved a lower bound for the stronger problem of spectral sparsification, which requires preserving all quadratic forms simultaneously rather than answering a single fixed query. They show that any incidence sketch that outputs even a constant-factor weighted spectral sparsifier requires $n^{21/20-o(1)}$ measurements. Our sketch is also incidence-based, as are many natural graph-streaming sketches. Thus, while their lower bound does not apply to spectral sketches in general, it suggests that obtaining spectral guarantees within the incidence-sketch framework may be substantially more difficult.

Our techniques also extend to the minimum $\ell_p$-norm $k$-partitioning
problem for every fixed $k\geq2$ and $p\in[1,\infty]$. In this problem,
the goal is to partition the vertex set into $k$ nonempty parts. Each
part has a boundary weight, given by the total weight of edges with
exactly one endpoint in that part, and the objective is to minimize the
$\ell_p$-norm of these $k$ boundary weights.
The case $p=1$ is equivalent to the minimum $k$-cut problem which has been studied
extensively in the classical setting
\cite{goldschmidt1994polynomial,thorup2008minimum,chekuri2020lp,li2019faster}.
The case $p=\infty$ is the Minmax $k$-cut problem, studied by
Chandrasekaran and Wang \cite{CW2021fixed}; their work also considers
the more general minimum $\ell_p$-norm $k$-partitioning objective for
every $p\geq1$.

A generic streaming approach is to construct a for-all cut sparsifier.
Since such a sparsifier preserves the boundary weight of every part, it
also preserves the $\ell_p$-objective of every $k$-partition. However, this again requires $\widetilde{\Omega}(n/\eps^2)$ bits of space.
By combining our for-each cut sketch with a constant-accuracy for-all sparsifier, we improve the dependence on $\eps$ to $\Ot(n/\eps)$ for every fixed $k$ and every $p\in[1,\infty]$.

\begin{result}[Formalized in \Cref{cor:kpartition}]
Fix a constant integer $k\geq2$ and let $p\in[1,\infty]$.
There is a randomized one-pass streaming
algorithm that, given a weighted graph under general
turnstile updates whose final edge weights are nonnegative, uses
$\Ot(n/\eps)$ bits of space and, with high probability, outputs a
$(1+\eps)$-approximate minimum $\ell_p$-norm $k$-partition. The update time is $t+ \poly \log n$ where $t$ is the update time of a constant-approximation for-all cut sparsifier.
\end{result}

Our techniques also give an improved sketch for evaluating weighted
correlation clusterings. In weighted correlation clustering, each pair
$\{u,v\}$ is associated with two nonnegative weights $w^+_{uv}$ and
$w^-_{uv}$, corresponding respectively to the cost of placing $u$ and
$v$ in different clusters and in the same cluster. The disagreement
cost of a clustering is the sum of these costs over all pairs.
Ahn, Cormode, Guha, McGregor, and Wirth~\cite{ACGMW2021correlation}
gave an $\Ot(1/\eps^2)$-space sketch for evaluating any fixed clustering
in the unweighted setting, and an $\Ot(n/\eps^2)$-space sketch for
evaluating any fixed two-clustering with arbitrary nonnegative weights.
As a direct consequence of our for-each cut sketch, we improve the
latter bound to $\Ot(n/\eps)$ bits, even under general turnstile updates
to the positive and negative weights. 
\begin{result}[Formalized in \Cref{cor:correlation-clustering}]
There is a randomized one-pass streaming algorithm that, given
general turnstile updates to the positive and negative weights
$\{w^+_{uv},w^-_{uv}\}$ of a weighted correlation-clustering
instance whose final weights are nonnegative, uses $\Ot(n/\eps)$
bits of space and, with high probability, estimates the
disagreement cost of any fixed two-clustering within a
$(1\pm\eps)$ factor. The update time is $\poly\log n$.
\end{result}

Our techniques also extend to for-each cut sketching for balanced directed graphs $G(V,E,W)$. A weighted directed graph is $\beta$-balanced if, for every nontrivial cut of the input graph $\calC\subset V$, the total weight of edges from $\calC$ to $V\setminus \calC$ is at most a factor $\beta$ larger than the total weight in the reverse direction. Cut sparsification for balanced directed graphs was previously studied outside the streaming setting. Ikeda and Tanigawa~\cite{IkedaT18} and Cen, Cheng, Panigrahi, and Sun~\cite{CCPS2021sparsification} gave offline constructions of cut sparsifiers and sketches; in particular, \cite{CCPS2021sparsification} gave an $\widetilde O(n\sqrt{\beta}/\eps)$-space for-each cut sketch. Cheng, Li, Lin, Tai, Woodruff, and Zhang~\cite{CLLTWZ2024tight} subsequently proved a matching $\widetilde\Omega(n\sqrt{\beta}/\eps)$ lower bound up to logarithmic factors for the corresponding data-structure problem. More recently, dynamic algorithms have been developed for balanced directed sparsification, maintaining for-all sparsifiers of size $\Ot(\beta n/\eps^2)$~\cite{GoranciHRS25,Zhao26}. Our construction achieves the optimal $\widetilde O(n\sqrt{\beta}/\eps)$ for-each space bound with a sketch that can be maintained in a single pass over a weighted directed graph under general turnstile updates. To the best of our knowledge, this is the first streaming algorithm to achieve the optimal for-each bound for balanced directed graphs. Moreover, the construction retains the same simple product-bucketing structure and streaming primitives as our undirected sketch.

\newpage

\begin{result}[Formalized in \Cref{thm:balanced-directed}]
For every $\beta\geq 1$, there is a randomized one-pass streaming algorithm that, given a weighted directed graph under general turnstile updates whose final edge weights are nonnegative and whose final graph is $\beta$-balanced, uses
$\Ot\left(\frac{n\sqrt{\beta}}{\eps}\right)$
bits of space and produces a for-each cut sketch. For any fixed nontrivial $\calC\subset V$, the sketch returns a $(1\pm\eps)$-approximation to the weighted cut value of $\calC$ with high probability. The update time is $\poly \log n$. 
\end{result}

\subsection{Technical Overview}

Our main tool for solving minimum cut is a for-each cut sketch. We use the same high-level reduction as in~\cite{ding2025}. We first store a constant-factor for-all cut sparsifier, say a $(1\pm 0.1)$-approximation, using $\Ot(n)$ space. This allows us to identify a polynomial-size family of sufficiently small cuts that contains the minimum cut, for example by enumerating all cuts whose value is within a constant factor of the minimum. We then evaluate every cut in this family using a for-each cut sketch with inverse-polynomial failure probability. A union bound guarantees that all candidate cut values are simultaneously preserved within a $(1\pm\eps)$ factor with high probability, and taking the minimum estimated value therefore yields a $(1+\eps)$-approximation to the minimum cut.

We now describe the construction of our $(1+\eps)$-approximate for-each cut sketch. For any cut $\calC$ fixed independently of the sketch randomness, it outputs, with high probability, an estimate of its weight $w$ with additive error at most $\eps w$.

\paragraph{A direct $L_1$-sketching baseline.}
A natural baseline is to estimate the cut weight directly using a linear $L_1$-norm sketch. Let $N=\binom{n}{2}$, with one coordinate for every unordered pair of vertices, and fix an arbitrary orientation of each pair. For each vertex $v$, define a signed incidence vector $x_v\in\mathbb{R}^N$: the coordinate corresponding to an edge $e$ incident to $v$ is $+w_e$ if $e$ is oriented out of $v$, $-w_e$ if it is oriented into $v$, and $0$ otherwise.

We maintain the same linear $L_1$-norm sketch of $x_v$ for every vertex $v$. Once the query cut $\calC$ is revealed, we add the sketches of all vertices in $\calC$. By linearity, this gives a sketch of
\[
x_{\calC}=\sum_{v\in\calC} x_v.
\]
Every edge with both endpoints in $\calC$ contributes through both of its endpoint vectors, once with sign $+w_e$ and once with sign $-w_e$, so the two contributions cancel. An edge with neither endpoint in $\calC$ contributes nothing to the sum. Finally, an edge crossing $\calC$ has exactly one endpoint in $\calC$, so it contributes exactly once, with value either $+w_e$ or $-w_e$. Thus, $x_{\calC}$ has one nonzero coordinate for each edge crossing $\calC$ and no other nonzero coordinates.
Therefore,
\[
\|x_{\calC}\|_1
= \sum_{|\{u,v\}\cap\calC|=1}w_{uv}
=
w.
\]

We can estimate this $L_1$ norm using a standard linear $L_1$-norm sketch.
A $(1\pm\eps)$ approximation requires
$\widetilde\Theta(1/\eps^2)$ space in the turnstile model~\cite{kanenw10a}.
Since we maintain such a sketch for every vertex (as we do not know in
advance of the stream what cut will be queried), this direct approach uses
$\widetilde O(n/\eps^2)$ space. Our goal is to save a factor of $1/\eps$
by exploiting additional information so that we do not need to estimate
the entire cut weight.

\paragraph{A bipartite view of cut edges.}
We create a bipartite graph $H=(V^L,V^R,E_H)$ by making two copies of
every original vertex $u\in V$: a left vertex $u^L\in V^L$ and a right
vertex $u^R\in V^R$. For every original edge $\{u,v\}$ with $u<v$, we
represent it by the single edge $(u^L,v^R)$ in $H$, with the same weight
as $\{u,v\}$. Thus, every original edge appears exactly once in the
bipartite graph.

Relative to the query cut $\calC$, assign each edge
$e=(u^L,v^R)$ the sign
\[
s_e =
\begin{cases}
+1, & u\in\calC,\ v\notin\calC,\\
-1, & u\notin\calC,\ v\in\calC,\\
0, & \text{otherwise}.
\end{cases}
\]
We call a crossing edge positive if $s_e=+1$ and negative if $s_e=-1$.
Thus, the nonzero signs correspond exactly to the edges crossing
$\calC$. We will exploit exact information about the difference between
the positive and negative cut-edge weights, rather than estimating the
entire cut weight directly. Figure~\ref{fig:graph_H} illustrates the bipartite representation and
previews the bucketing scheme introduced below.
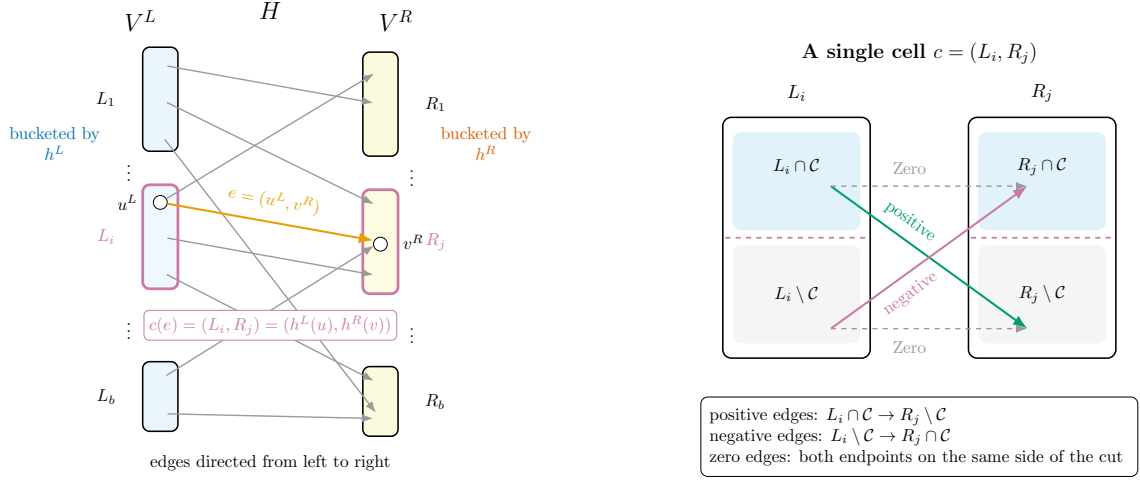
\begin{figure}[t] \label{fig:bipartite}
\centering

\begin{subfigure}[t]{0.48\textwidth}
\centering
\scalebox{0.65}{\begin{tikzpicture}[
    x=0.8cm,
    y=1cm,
    >=Latex,
    font=\small,
    vertex/.style={circle, draw=colorEight, fill=white, inner sep=1.6pt, minimum size=8pt},
    edge/.style={->, very thick, color=colorFive},
    softedge/.style={->, thick, color=colorNine},
    bucket/.style={rounded corners=4pt, draw=colorEight, line width=0.9pt, inner sep=6pt},
    highlight/.style={rounded corners=4pt, draw=colorSeven, line width=1.6pt, inner sep=6pt}
]

\node[font=\bfseries] at (-0.5,6.3) {\Large $V^L$};
\node[font=\bfseries] at (6,6.3) {\Large $V^R$};
\node[font=\bfseries\large] at (2.8,6.5) {\Large $H$};
\node[align=center] at (2.8,-2.7) {edges directed from left to right};

\node[vertex] (l1a) at (0,5.4) {};
\node[vertex] (l1b) at (0,4.7) {};
\node[vertex] (l1c) at (0,4.0) {};

\node[vertex] (lia) at (0,2.6) {};
\node[vertex] (lib) at (0,1.9) {};
\node[vertex] (lic) at (0,1.2) {};

\node[vertex] (lba) at (0,-1.0) {};
\node[vertex] (lbb) at (0,-1.7) {};

\node[vertex] (r1a) at (5.6,5.3) {};
\node[vertex] (r1b) at (5.6,4.6) {};
\node[vertex] (r1c) at (5.6,3.9) {};

\node[vertex] (rja) at (5.6,2.5) {};
\node[vertex] (rjb) at (5.6,1.8) {};
\node[vertex] (rjc) at (5.6,1.1) {};

\node[vertex] (rba) at (5.6,-1.1) {};
\node[vertex] (rbb) at (5.6,-1.8) {};

\node[bucket, fill=colorTwo!12, fit=(l1a)(l1b)(l1c)] (Lone) {};
\node[highlight, fill=colorTwo!10, fit=(lia)(lib)(lic)] (Li) {};
\node[bucket, fill=colorTwo!12, fit=(lba)(lbb)] (Lb) {};

\node[bucket, fill=colorFour!18, fit=(r1a)(r1b)(r1c)] (Rone) {};
\node[highlight, fill=colorFour!16, fit=(rja)(rjb)(rjc)] (Rj) {};
\node[bucket, fill=colorFour!18, fit=(rba)(rbb)] (Rb) {};

\node[vertex] (uL) at (0,2.6) {};
\node[vertex] (vR) at (5.6,1.75) {};

\node[left=4mm of Lone.west] {$L_1$};
\node[left=4mm of Li.west, text=colorSeven] {$L_i$};
\node[left=4mm of Lb.west] {$L_b$};

\node[right=4mm of Rone.east] {$R_1$};
\node[right=4mm of Rj.east, text=colorSeven] {$R_j$};
\node[right=4mm of Rb.east] {$R_b$};

\node at ($(Lone.south)!0.5!(Li.north)+(-0.8,0)$) {$\vdots$};
\node at ($(Li.south)!0.5!(Lb.north)+(-0.8,0)$) {$\vdots$};
\node at ($(Rone.south)!0.5!(Rj.north)+(0.8,0)$) {$\vdots$};
\node at ($(Rj.south)!0.5!(Rb.north)+(0.8,0)$) {$\vdots$};

\node[left=1.8mm of lia] {$u^L$};
\node[right=1.8mm of rjb] {$v^R$};

\draw[softedge] (l1a) -- (r1b);
\draw[softedge] (l1b) -- (rja);
\draw[softedge] (l1c) -- (rbb);
\draw[softedge] (lia) -- (r1a);
\draw[softedge] (lib) -- (rjc);
\draw[softedge] (lic) -- (rba);
\draw[softedge] (lba) -- (rjb);
\draw[softedge] (lbb) -- (rbb);

\draw[edge, color=colorOne] (lia) -- node[midway, above, sloped, yshift=3pt] {$e=(u^L,v^R)$} (rjb);
\node[draw=colorSeven, rounded corners=3pt, fill=white, inner sep=3pt, text=colorSeven]
    at (2.8,0.1) {$c(e)=(L_i,R_j)=(h^L(u),h^R(v))$};

\node[color=colorFive, align=center, anchor=east, xshift=-6pt]
    at (-1.1,3.8) {bucketed by\\$h^L$};

\node[color=colorSix, align=center, anchor=west, xshift=6pt]
    at (6.7,3.8) {bucketed by\\$h^R$};



\end{tikzpicture}}
\caption{The bipartite representation $H$, partitioned into left buckets $L_i$ by $h^L$ and right buckets $R_j$ by $h^R$. Each edge $\{u,v\}$ with $u<v$ is represented by the directed edge $(u^L,v^R)$.}
\end{subfigure}
\hfill
\begin{subfigure}[t]{0.48\textwidth}
\centering
\scalebox{0.65}{\begin{tikzpicture}[
    x=1cm,
    y=1cm,
    >=Latex,
    font=\small,
    vertex/.style={coordinate},
    posedge/.style={->, very thick, color=colorThree},
    negedge/.style={->, very thick, color=colorSeven},
    zeroedge/.style={->, thick, dashed, color=colorNine},
    bucket/.style={rounded corners=5pt, draw=colorEight, line width=0.9pt, inner sep=7pt}
]

\node[font=\bfseries\large] at (3.6,6.7) {A single cell $c=(L_i,R_j)$};

\coordinate (LtopNW) at (-0.2,5.1);
\coordinate (LtopSE) at (2.4,3.1);
\coordinate (LbotNW) at (-0.2,2.8);
\coordinate (LbotSE) at (2.4,0.8);

\coordinate (RtopNW) at (4.8,5.1);
\coordinate (RtopSE) at (7.4,3.1);
\coordinate (RbotNW) at (4.8,2.8);
\coordinate (RbotSE) at (7.4,0.8);

\fill[colorTwo!18, rounded corners=5pt] (LtopNW) rectangle (LtopSE);
\fill[colorNine!10, rounded corners=5pt] (LbotNW) rectangle (LbotSE);
\fill[colorTwo!18, rounded corners=5pt] (RtopNW) rectangle (RtopSE);
\fill[colorNine!10, rounded corners=5pt] (RbotNW) rectangle (RbotSE);

\draw[bucket] (-0.4,0.5) rectangle (2.6,5.4);
\draw[bucket] (4.6,0.5) rectangle (7.6,5.4);

\draw[dashed, line width=1pt, color=colorSeven] (-0.35,2.95) -- (2.55,2.95);
\draw[dashed, line width=1pt, color=colorSeven] (4.65,2.95) -- (7.55,2.95);

\node[font=\bfseries] at (1.1,5.9) {$L_i$};
\node[font=\bfseries] at (6.1,5.9) {$R_j$};

\node at (1.1,4.4) {$L_i\cap \mathcal C$};
\node at (1.1,1.8) {$L_i\setminus \mathcal C$};
\node at (6.1,4.4) {$R_j\cap \mathcal C$};
\node at (6.1,1.8) {$R_j\setminus \mathcal C$};

\node[left=6pt, text=colorSix] at (-0.4,4.4) {};
\node[left=6pt, text=colorSix] at (-0.4,1.8) {};
\node[right=6pt, text=colorSix] at (7.6,4.4) {};
\node[right=6pt, text=colorSix] at (7.6,1.8) {};

\node[vertex] (lt2) at (1.8,4.0) {};
\node[vertex] (lb2) at (1.8,1.1) {};

\node[vertex] (rt2) at (5.8,4) {};
\node[vertex] (rb2) at (5.8,1.1) {};

\draw[posedge] (lt2) -- node[pos=0.35, above, sloped, yshift=2pt, text=colorThree] {positive} (rb2);
\draw[negedge] (lb2) -- node[pos=0.35, below, sloped, yshift=-2pt, text=colorSeven] {negative} (rt2);
\draw[zeroedge] (lt2) -- node[pos=0.4, above, sloped, yshift=3pt, text=colorNine] {Zero} (rt2);
\draw[zeroedge] (lb2) -- node[pos=0.4, below, sloped, yshift=-3pt, text=colorNine] {Zero} (rb2);


\node[draw=colorEight, rounded corners=4pt, fill=white, inner sep=5pt, align=left] at (3.6,-1.1) {
    \begin{tabular}{@{}l@{}}
    positive edges: $L_i\cap \mathcal C \to R_j\setminus \mathcal C$\\
    negative edges: $L_i\setminus \mathcal C \to R_j\cap \mathcal C$\\
    zero edges: both endpoints on the same side of the cut
    \end{tabular}
};

\end{tikzpicture}}
\caption{The four edge types in a cell $c=(L_i,R_j)$ induced by the cut $\calC$. Crossing edges have sign ($s_e$) $+1$ or $-1$, while noncrossing edges have sign $0$.}
\end{subfigure}
\caption{Figure showing the graph $H$ and a cut $\calC$.}
\label{fig:graph_H}
\end{figure}

\paragraph{Ideal thought experiment: independently throw edges into $B^2$ cells.} For a parameter \(B\) to be chosen later,
assign every edge independently and uniformly to one of $B^2$ cells (before the query cut is known). These cells partition the edges, not the
vertices. For a cell $c$, let $m_c^+$ and $m_c^-$ denote the total weight
of the positive and negative cut edges assigned to $c$, respectively, and
let
\[
m_c^{\mathrm{tot}} = m_c^+ + m_c^-
\]
denote the total cut weight contributed by the cell.

For every cell $c$ and every left vertex $u^L$, we store the total weight
of edges in $c$ incident to $u^L$. Similarly, for every right vertex
$v^R$, we store the total weight of edges in $c$ incident to $v^R$.
Once the query cut $\calC$ is revealed, we add the counters for all
left vertices $u^L$ with $u\in\calC$ and subtract the counters for all
right vertices $v^R$ with $v\in\calC$.

Consider an edge $(u^L,v^R)$ in cell $c$. If both $u,v\in\calC$, then
its weight appears once in each sum and the two contributions cancel.
If neither $u$ nor $v$ belongs to $\calC$, the edge appears in neither
sum. If $u\in\calC$ and $v\notin\calC$, the edge contributes $+w_{uv}$,
while if $u\notin\calC$ and $v\in\calC$, it contributes $-w_{uv}$.
Thus, the resulting signed sum is exactly $m_c^+-m_c^-$.
The crucial identity is
\[
\boxed{
m_c^{\mathrm{tot}}
=
m_c^+ + m_c^-
=
|m_c^+-m_c^-|
+
2\min\{m_c^+,m_c^-\}.
}
\]

The absolute imbalance $|m_c^+-m_c^-|$ is therefore known exactly, and only the minority
mass $\min\{m_c^+,m_c^-\}$ needs estimation. This identity also shows where the difficulty
lies. If one sign dominates the cell, then the imbalance already accounts for
most of the cell's cut weight, leaving only a small minority mass to estimate.
The difficult case is when the positive and negative masses are comparable:
substantial weight then cancels in the signed difference, leaving a large
minority mass to estimate. We estimate precisely this
hidden mass using an $L_1$-norm estimator together with an independent
$\ell_1$-sampler.

\paragraph{Estimating only the minority.}
To estimate the minority mass, we maintain two sketches for each vertex--cell pair. Specifically, for every cell $c$ and every left or right vertex incident to edges in $c$, we maintain an $L_1$-norm sketch and an $\ell_1$-sampler over the edges of $c$ incident to that vertex. For each cell, corresponding sketches across vertices use the same linear maps, while the randomness of the $L_1$-norm sketch and the $\ell_1$-sampler is independent.

Once the query cut $\calC$ is revealed, we combine these sketches using the same signed sum as for the exact counters above. By linearity, the resulting $L_1$-norm sketch corresponds to a vector whose nonzero coordinates are exactly the cut edges in the cell, with coordinate $+w_e$ for a positive edge and $-w_e$ for a negative edge. Its $L_1$ norm is therefore
$m_c^{\mathrm{tot}}=m_c^++m_c^-$.
Using a constant-dimensional linear sketch based on Cauchy random projections, we obtain an unbiased estimator $A_c$ of this quantity satisfying
\[
\mathbb{E}[A_c] = m_c^{\mathrm{tot}}
\qquad\text{and}\qquad
\mathbb{E}[A_c^2] = O\!\left((m_c^{\mathrm{tot}})^2\right).
\]

Independently, the combined $\ell_1$-sampler returns a cut edge $e$ from the cell with probability proportional to its weight. The exact imbalance $m_c^+-m_c^-$ tells us which sign is the minority. Let
$m_c^{\min}=\min\{m_c^+,m_c^-\},$
and let $F_c=1$ if the sampled edge has the minority sign in cell $c$, and $F_c=0$ otherwise. Since the sampler selects an edge with probability proportional to its weight,
\[
\Pr[F_c=1]
=
\frac{m_c^{\min}}{m_c^{\mathrm{tot}}}.
\]
By independence of $A_c$ and $F_c$,
\[
\mathbb{E}[A_cF_c]
=
\mathbb{E}[A_c]\Pr[F_c=1]
=
m_c^{\mathrm{tot}}
\frac{m_c^{\min}}{m_c^{\mathrm{tot}}}
=
m_c^{\min}.
\]
Thus, $A_cF_c$ is an unbiased estimate of the minority mass $m_c^{\min}$, and $2A_cF_c$ is an unbiased estimate of the correction term $2m_c^{\min}$ in the identity above.

\paragraph{What controls the variance?}
Suppose the positive edges are the majority in total weight in cell $c$,
so the negative edges form the minority. The estimator uses $2A_cF_c$
to estimate the correction term $2m_c^-$, where $F_c=1$ exactly when
the $\ell_1$-sampler returns a negative edge, and $F_c=0$ otherwise.
Define the cell estimator
\[
\widehat m_c^{\mathrm{tot}}
=
|m_c^+-m_c^-|+2A_cF_c.
\]
Since the first term is deterministic once the cell is fixed,
$\operatorname{Var}\!\left(\widehat m_c^{\mathrm{tot}}\right)
=
4\,\operatorname{Var}(A_cF_c).$

To bound $\operatorname{Var}(A_cF_c)$, we use
$\operatorname{Var}(A_cF_c)
\leq
\mathbb{E}\!\left[(A_cF_c)^2\right].$
Since $F_c\in\{0,1\}$ and $A_c$ and $F_c$ are independent,
\[
\mathbb{E}\!\left[(A_cF_c)^2\right]
=
\mathbb{E}[A_c^2]\Pr[F_c=1].
\]
Using
$\mathbb{E}[A_c^2]
=
O\!\left((m_c^{\mathrm{tot}})^2\right)$
and
$\Pr[F_c=1]
=
\frac{m_c^-}{m_c^{\mathrm{tot}}},$
we obtain
\[
    \operatorname{Var}
    \bigl(\widehat m_c^{\mathrm{tot}}\bigr)
    \lesssim
    (m_c^{\mathrm{tot}})^2 \cdot \frac{ m_c^-}{m_c^\mathrm{tot}} =
    m_c^{\mathrm{tot}}m_c^-
    \lesssim
    m_c^+m_c^-,
\]
where the last inequality uses
$m_c^-\leq m_c^+$ and hence
$m_c^{\mathrm{tot}}=m_c^++m_c^-\leq 2m_c^+$.

The variance is small in two favorable cases. If the cell has small total mass, then the $L_1$-norm estimate itself is small in scale. If the cell is nearly one-sided, then the exact imbalance already captures almost all of its cut weight, leaving only a small minority mass to estimate. The difficult case is therefore a large cell containing substantial mass of both signs: much of its cut weight cancels in the signed imbalance, leaving a large minority mass that must be estimated. Importantly, we only need the sum of the cell estimates to have small additive error, rather than requiring a relative approximation for every individual cell.

\paragraph{Bucketing suppresses opposite-sign interaction.}
The product $m_c^+m_c^-$ is the total weighted contribution of
positive--negative edge pairs sharing cell $c$, where each pair
contributes the product of its two edge weights. Over the whole cut,
the total contribution of all positive--negative pairs is
$m^+m^-
\leq
\frac{w^2}{4},$
where $m^+$ and $m^-$ are the total positive and negative cut-edge
weights and $w=m^++m^-$. Since each positive--negative pair is assigned
to the same cell with probability $1/B^2$,
\[
\mathbb{E}_{\mathrm{bucket}}
\left[\sum_c m_c^+m_c^-\right]
=
\frac{m^+m^-}{B^2}
\leq
\frac{w^2}{4B^2}.
\]
Thus, bucketing suppresses precisely the interactions that contribute to
the variance: a cell has large variance only when it contains substantial
weight from both positive and negative edges.

Our final estimate $\widehat w$ is the sum of the cell estimates. For
every fixed bucket assignment, the cell estimates are unbiased and use
independent estimation randomness, so their variances add. By the per-cell variance bound above,
\[
\operatorname{Var}
\bigl(\widehat w\mid \text{bucket assignment}\bigr)
\lesssim
\sum_c m_c^+m_c^-.
\]
Moreover,
\[
\mathbb{E}
\bigl[\widehat w\mid \text{bucket assignment}\bigr]
=
w
\]
for every bucket assignment. Thus, by the law of total variance,
\[
\begin{aligned}
\operatorname{Var}(\widehat w)
&=
\mathbb{E}_{\mathrm{bucket}}
\left[
\operatorname{Var}
\bigl(\widehat w\mid \text{bucket assignment}\bigr)
\right]
+
\operatorname{Var}_{\mathrm{bucket}}
\left(
\mathbb{E}
\bigl[\widehat w\mid \text{bucket assignment}\bigr]
\right) \\
&=
\mathbb{E}_{\mathrm{bucket}}
\left[
\operatorname{Var}
\bigl(\widehat w\mid \text{bucket assignment}\bigr)
\right] \\
&\lesssim
\mathbb{E}_{\mathrm{bucket}}
\left[\sum_c m_c^+m_c^-\right]
=
O\left(\frac{w^2}{B^2}\right).
\end{aligned}
\]

We choose $B=\Theta(1/\eps)$, with a sufficiently large constant, so that
$\operatorname{Var}(\widehat w)
\leq
\frac{\eps^2w^2}{10}.$
Since $\mathbb{E}[\widehat w]=w$, Chebyshev's inequality gives
$\Pr\left[
|\widehat w-w|>\eps w
\right]
\leq
\frac{1}{10}.$
Thus, one repetition returns a $(1\pm\eps)$ approximation to the cut
weight with constant probability. Independent repetitions and taking the
median amplify this guarantee to high probability.

\paragraph{Why the ideal implementation still costs $\Ot(n/\eps^2)$.}
In the ideal implementation above, each vertex--cell pair stores an exact
counter, an $L_1$-norm sketch, and an $\ell_1$-sampler, requiring only
polylogarithmic space per pair. However, because edges are assigned to
cells independently of their endpoints, edges incident to a single
vertex can be scattered across all $B^2$ cells. Thus, each of the $2n$
vertex copies may require $B^2$ summaries, for a total of $2nB^2$
vertex--cell pairs. The resulting space complexity is
\[
\Ot(nB^2)=\Ot(n/\eps^2),
\]
where we use $B=\Theta(1/\eps)$. Thus, independently partitioning the
edges into $B^2$ cells gives the variance reduction we need, but storing
the corresponding summaries still costs too much space.

\paragraph{Product bucketing: preserve $B^2$ cells, pay for only $B$.}
Instead, we independently assign each left vertex $u^L$ to one of $B$
left buckets $L_1,\ldots,L_B$ and each right vertex $u^R$ to one of $B$
right buckets $R_1,\ldots,R_B$, uniformly at random. A cell
$c=(L_i,R_j)$ consists of all edges whose left endpoint lies in $L_i$
and whose right endpoint lies in $R_j$.
Figure~\ref{fig:graph_H} illustrates this product-bucketing construction
and the positive and negative edges within a cell.

There are still $B^2$ cells, but a left vertex in $L_i$ can be incident
only to cells $(L_i,R_1),\ldots,(L_i,R_B)$. It therefore needs only $B$
summaries, one for each right bucket. Symmetrically, each right vertex
needs only $B$ summaries, one for each left bucket. To query a cell
$(L_i,R_j)$, we combine the corresponding left- and right-vertex
summaries using the same signed sums as before. Thus, there are only
$2nB$ vertex--bucket summaries in total, reducing the space to
$\Ot(nB)=\Ot(n/\eps).$

Moreover, an update to an edge affects only the two vertex--bucket
structures associated with its endpoints. Since the exact counters,
$L_1$-norm sketches, and $\ell_1$-samplers can each be updated in
polylogarithmic time, the resulting worst-case update time is
$\poly \log(n)$.

\paragraph{Why dependencies between edge locations do not hurt.}
The storage reduction introduces dependencies between edge locations:
all edges incident to a left vertex share a row, and all edges incident
to a right vertex share a column. A high-degree or high-weight vertex
can therefore concentrate substantial weight in a small number of cells, and we cannot assume that the cut weight is evenly distributed across
cells.

Crucially, however, the variance analysis only uses the probability that
a fixed positive--negative pair of cut edges lands in the same cell.
The left endpoint of the positive edge corresponds to a vertex in
$\calC$, while the left endpoint of the negative edge corresponds to a
vertex outside $\calC$, so these two endpoints are distinct. Similarly,
their right endpoints are distinct. Hence the two left endpoints land in
the same left bucket with probability $1/B$, and the two right endpoints
land in the same right bucket with probability $1/B$. Since the left and
right bucket assignments are independent, the two edges share a cell
with probability $1/B^2$.

By linearity of expectation, we therefore retain the same bound on
$\mathbb{E}\left[\sum_c m_c^+m_c^-\right],$
even though the locations of different edge pairs may be correlated.
Thus, the same variance analysis applies. The central saving is
\emph{the statistical benefit of $B^2$ cells at the storage cost of only
$B$ summaries per vertex}.
\subsection{Related Work}

Exact minimum cut admits space-efficient algorithms when additional
passes are allowed. Although developed in the cut query model,
the $\Ot(n)$-query algorithm of Rubinstein, Schramm, and
Weinberg~\cite{rubinsteinsw18} also implies a two-pass streaming
algorithm using $\Ot(n)$ space for simple unweighted graphs.
Assadi and Dudeja~\cite{assadi2021simple} formalized this implication
and gave a simpler two-pass streaming algorithm for exact minimum
cut in such graphs using $O(n\log n)$ bits of space.
For weighted graphs, Mukhopadhyay and
Nanongkai~\cite{mukhopadhyayn20} obtained an exact algorithm
using $\Ot(n)$ space and $O(\log n)$ passes.
In random order streams, exact minimum cut in simple unweighted
graphs can be computed in one pass using $\Ot(n)$
space~\cite{ding2025}.
For approximate minimum cut in insertion-only streams,
Cohen-Addad, Woodruff, Xie, and Zhou~\cite{cohenaddad2026nearly}
improved the logarithmic factors in the space bound
of~\cite{ding2025}.

Beyond minimum cut, $k$-edge-connectivity and $k$-vertex-connectivity
have also been studied in graph streams. These works address
both connectivity testing and sparse certificate
construction~\cite{ahngm12,guha2015vertex,assadi2023tight,
ghosh2024new,sawettamalya2026nearly}.
For $k=1$, Nelson and Yu~\cite{nelson2019optimal} proved optimal
space lower bounds for computing spanning forests in dynamic streams.

\section{Preliminaries}

\paragraph{Basic notation.}
We use $[n]$ to denote the set $\{1,2,\ldots,n\}$.
All logarithms are base $2$. 
We take with high probability to mean with probability $1-1/\poly(n)$. We may assume without loss of generality that $\eps\geq 1/n$,
since otherwise we can store the entire weighted graph explicitly
using $\Ot(n^2)=\Ot(n/\eps)$ bits. For $x\in\mathbb{R}^d$, we write
$\|x\|_1:=\sum_{i=1}^d |x_i|$.

\subsection{Streaming model}

The vertex set $V=[n]$ is fixed in advance, and the graph is represented by a vector $x\in \mathbb{Z}^{\binom{V}{2}}$, where each coordinate $x_e$ denotes the weight of an edge $e\in \binom{V}{2}$. The input is a stream of updates $(e,\Delta)$, where $\Delta\in\mathbb{Z}$, which perform $x_e \leftarrow x_e+\Delta$. Positive updates increase the weight of an edge and negative updates decrease it. We assume that $x_e\geq 0$ at the end of the stream for every edge $e$. As is standard in the literature, we further assume that both the stream length and the magnitude of each update $\Delta$ are bounded by $\poly(n)$. We consider one-pass streaming algorithms, with space complexity measured in bits. As is standard for dynamic graph streaming algorithms, our algorithms are based on linear sketches \cite{linw14,jiang2026turnstile}.

\subsection{Cuts}
For a simple, undirected graph $G=(V,E)$ let a cut of $G$ be $\calC \subsetneq V$ such that $\calC \neq \emptyset$.
Let the edges crossing the cut $\calC$ be 
\[
\partial_G(\calC) = \{\{u,v\} \in E: |\{u,v\} \cap \calC| = 1\}. 
\]
We denote the global minimum-cut value by $\lambda(G):=\min_{\emptyset\subsetneq\calC\subsetneq V}|\partial_G(\calC)|$.
The minimum-cut problem asks for a cut $\calC$ attaining this value; its $(1+\eps)$-approximate version asks for a cut $\calC$ satisfying $|\partial_G(\calC)|\leq(1+\eps)\lambda(G)$.

For a weighted, undirected graph $G=(V,E,w)$, a cut is a set $\calC\subsetneq V$ such that $\calC\neq\emptyset$. We denote the weight of the cut by
\[
w_G(\calC):=\sum_{\substack{\{u,v\}\in E\ \\ |\{u,v\}\cap\calC|=1}} w_{uv},
\]
and define the global minimum-cut value as
$\lambda(G):=\min_{\emptyset\subsetneq\calC\subsetneq V} w_G(\calC).$
The $(1+\eps)$-approximate min-cut problem asks for a cut $\calC$ satisfying
$w_G(\calC)\le (1+\eps)\lambda(G).$

For a partition $\mathcal P=(\calC_1,\ldots,\calC_k)$ of $V$ into $k$ nonempty parts, define the weight of the corresponding \emph{$k$-cut} by
\[
w_G(\mathcal P)
:=
\sum_{\substack{\{u,v\}\in E\\
u\in \calC_i,\ v\in \calC_j,\ i\neq j}}
w_{uv}.
\]
We denote the minimum $k$-cut value by $\lambda_k(G)
:=
\min_{\mathcal P} w_G(\mathcal P),$
where the minimum is over all partitions of $V$ into $k$ nonempty parts. The $(1+\eps)$-approximate minimum $k$-cut problem asks for a partition
$\mathcal P$ of $V$ into $k$ nonempty parts satisfying
$w_G(\mathcal P)\leq (1+\eps)\lambda_k(G).$

\subsection{\texorpdfstring{$\ell_1$}{L1} samplers} 

For $x \in \R^d$, an exact  $\ell_1$\textbf{-sampler} outputs an index $i \in \{1,2,\dots, d\}$ s.t. for every $j \in \{1,2,\dots, d\}$ we have 
$\Pr[i=j]
    =
    \frac{|x_j|}{\|x\|_1}.$

    \begin{proposition}[\cite{SWZ25perfect}]
\label{lem:l1Samp}
Let $x\in\{-\poly(d),\ldots,\poly(d)\}^d$ be presented by a turnstile
stream, and let $\delta\in(0,1)$. For any constant $C>0$, there is a
one-pass $\ell_1$-sampler linear sketch that fails with probability at most $\delta$
and, conditioned on not failing, outputs an index $i\in[d]$ such that,
for every $j\in[d]$,
\[
\Pr[i=j]
=
\frac{|x_j|}{\|x\|_1}
\pm
\frac{1}{d^C}.
\]
The sampler uses
$\widetilde O(\log^2 d\log(1/\delta))$ bits of space and
$\poly \log(d)$ time per update.
\end{proposition}

In this paper, $d$ will be $\binom{n}{2}$. Since the stream has polynomial length and the algorithm uses only polynomially many sampler outputs, for sufficiently large $C$ the joint distribution is within $n^{-c}$ total variation distance of that obtained with exact $\ell_1$-samplers for any constant $c > 0$. We therefore treat the samplers as exact throughout the analysis and absorb this negligible discrepancy into the final failure probability.

\subsection{Cut Sparsifiers}

A weighted graph $H$ on the same vertex set is a $(1\pm\eps)$ for-all cut
sparsifier of $G$ if it approximates the value of every cut simultaneously.
We use the following streaming construction.

\begin{proposition}[\cite{CKL2022weighted}]
\label{thm:weighted-for-all-sparsifier}
For every $\eps\in(0,1)$, there is a randomized one-pass algorithm in the
general weighted turnstile model that uses $\Ot(n/\eps^3)$ bits of space and,
with high probability, recovers a weighted graph $H$ satisfying
$(1-\eps)w_G(\calC)
    \leq w_H(\calC)
    \leq (1+\eps)w_G(\calC)
$
simultaneously for every nontrivial cut $\calC\subsetneq V$ of the final
graph $G$.
\end{proposition}

In particular, setting $\eps$ to a constant gives a for-all cut sparsifier in $\Ot(n)$ bits of space in the same update model.

\subsection{$L_1$ norm sketch via Cauchy random variables}
Here we use standard techniques to build an $L_1$-norm sketch with the
guarantees required by our algorithm. 

\begin{definition}
A random variable $Z$ is a standard Cauchy random variable if it has
probability density function
$f_Z(z)=\frac{1}{\pi(1+z^2)}$ for $z\in\mathbb{R}.$
\end{definition}

\begin{proposition}[Cauchy $1$-stability; cf.~\cite{Indyk2006}]
\label{lem:cauchy-stability}
Let $Z_1,\ldots,Z_d$ be independent standard Cauchy random variables,
and let $x\in\mathbb{R}^d$ be fixed. Then
$\sum_{j=1}^d x_j Z_j
\stackrel{d}{=}
\|x\|_1 Z,$
where $Z$ is a standard Cauchy random variable.
\end{proposition}

\begin{lemma}
\label{lem:cauchy-l1-mass}
Let $x\in\mathbb{R}^d$ be any fixed vector, and fix an integer $r>2$.
For every $i\in[r]$ and $j\in[d]$, let $\rho_{i,j}$ be independent
standard Cauchy random variables, and define
$X_i=\sum_{j=1}^d \rho_{i,j}x_j
\qquad\text{for }i\in[r].$
Define
\[
A_r(x)
=
\cos^r\!\left(\frac{\pi}{2r}\right)
\prod_{i=1}^r |X_i|^{1/r}.
\]
Then
\[
\E[A_r(x)]=\|x\|_1,
\qquad
\E[A_r(x)^2]=C_r\|x\|_1^2,
\]
where
$C_r
:=
\frac{
\cos^{2r}\!\left(\frac{\pi}{2r}\right)
}{
\cos^r\!\left(\frac{\pi}{r}\right)
}.
$
In particular,
$\Var(A_r(x))
=
(C_r-1)\|x\|_1^2
=
O_r(\|x\|_1^2).$
\end{lemma}

\begin{proof}
If $x=0$, the result is immediate. By
\Cref{lem:cauchy-stability},
$X_i\stackrel{d}{=}\|x\|_1 Z_i,$
where $Z_1,\ldots,Z_r$ are independent standard Cauchy random variables.
The result then follows directly from
\cite[Lemma~6]{LiHastieChurch2007}.
\end{proof}

The Cauchy sketch above is defined using exact real arithmetic and fully
independent Cauchy random variables. To obtain a finite-space streaming
implementation, we use the bounded-precision and pseudorandomness
techniques of~\cite{Indyk2006}. At a high level, the Cauchy coefficients
and maintained measurements can be represented with logarithmically many
bits, and the coefficients can be generated as needed from a short
pseudorandom seed. We defer the formal statements and proofs to
\Cref{appen:prelim}; the following consequence is all that we use.

\begin{lemma}
\label{lem:cauchy-sketch}
Fix a constant integer $r>2$. For any fixed constant $C>0$, there is a
linear sketch for
$x\in\{-\poly(d),\ldots,\poly(d)\}^d$
that maintains an estimator $\widetilde A_r(x)$ using $\poly \log(d)$ bits
of space and $\poly \log(d)$ update time, such that
$\Pr\left[
|\widetilde A_r(x)-A_r(x)|
>
d^{-C}\|x\|_1
\right]
\leq d^{-C}.$
\end{lemma}

Throughout the analysis, we work with the ideal Cauchy estimator
$A_r(x)$ of \Cref{lem:cauchy-l1-mass}. By
\Cref{lem:cauchy-sketch}, and since our algorithm uses only
polynomially many estimator blocks, choosing the constant $C$
sufficiently large makes the cumulative finite-precision error and
failure probability negligible. We absorb these errors into the final
approximation and failure probability. We take $r=10$, for which
$C_r<4/3$. We refer to the corresponding Cauchy measurements
$X_1,\ldots,X_r$, together with the decoder $A_r$, as an
$L_1$\textbf{-norm estimator block.}

\subsection{Probabilistic Tools}

We use the following standard concentration inequalities.

\begin{proposition}[Markov's inequality; cf.~\cite{dubhashi2009concentration}]
\label{prop:markov}
Let $X$ be a nonnegative random variable. Then, for every $t>0$,
\[
\Pr[X\geq t]
\leq
\frac{\E[X]}{t}.
\]
\end{proposition}

\begin{proposition}[Chebyshev's inequality; cf.~\cite{dubhashi2009concentration}]
\label{prop:chebyshev}
Let $X$ be a random variable with finite variance. Then, for every $t>0$,
\[
\Pr\!\left[|X-\E[X]|\geq t\right]
\leq
\frac{\Var(X)}{t^2}.
\]
\end{proposition}

 \begin{proposition}[Chernoff bound; cf.~\cite{dubhashi2009concentration}]\label{prop:chernoff}
 	Suppose $X_1,\ldots,X_m$ are independent random variables taking values in $[0,1]$. Let $X := \sum_{i=1}^m X_i$ and $\mu_L \leq \expect{X} \leq \mu_H$. Then, for any $\eta > 0$, 
 	\[
 	\Pr\paren{X >  (1+\eta) \cdot \mu_H} \leq \exp\paren{-\frac{\eta^2 \cdot \mu_H}{3+\eta}} \quad \textnormal{and} \quad \Pr\paren{X <  (1-\eta) \cdot \mu_L} \leq \exp\paren{-\frac{\eta^2 \cdot \mu_L}{2+\eta}}.
 	\]
 \end{proposition}

 We also use limited-independence hash functions in our algorithm.

\begin{definition}
\label{def:hash}
For integers $n,m,k\geq 1$, a family
$\mathcal H\subseteq\{h:[n]\to[m]\}$ is $k$-wise independent if,
for every $t\leq \min\{k,n\}$, every collection of distinct inputs
$a_1,\ldots,a_t\in[n]$, and every choice of outputs
$b_1,\ldots,b_t\in[m]$,
\[
    \Pr_{h\sim\mathcal H}
    \bigl[h(a_1)=b_1,\ldots,h(a_t)=b_t\bigr]
    =
    m^{-t}.
\]
The output values $b_1,\ldots,b_t$ need not be distinct.
\end{definition}

Roughly speaking, a $k$-wise independent hash function behaves like a
fully random function on any set of at most $k$ distinct inputs. We use
the following standard construction.

\begin{proposition}[\cite{MotwaniR95}]
\label{prop:k-wise}
For all integers $n,m,k\geq 2$, there exists a $k$-wise independent
family
$\mathcal H\subseteq\{h:[n]\to[m]\}
$
such that a uniformly random function $h\in\mathcal H$ can be sampled
and stored using
$O\bigl(k(\log n+\log m)\bigr)
$
bits of space.
\end{proposition}

We also use the following Chernoff--Hoeffding-type concentration bound
for limited independence.

\begin{proposition}[\cite{schmidtss95}]
\label{prop:chernoff-limited}
Suppose $h:[n]\to[m]$ is sampled from a $k$-wise independent family,
where $n,m\geq1$ and $k\geq2$.
Let $u_1,\ldots,u_s\in[n]$ be distinct, where $0\leq s\leq n$, and let
$X_1,\ldots,X_s$ be random variables in $\set{0,1}$, where
$X_i=1$ if and only if $h(u_i)=1$.
Set $X:=\sum_{i=1}^s X_i$ and $\mu:=\E[X]$.
If $\mu>0$, then, for every $a>0$,
\[
    \Pr\paren{|X-\mu|\geq a}
    \leq
    \exp\!\left(
        -\left\lfloor
            \min\left\{
                \frac{k}{2},\,
                \frac{a^2}{3\mu},\,
                \frac{a}{3}
            \right\}
        \right\rfloor
    \right).
\]
If $\mu=0$, then $\Pr[X=0]=1$.
\end{proposition}

\section{Min-cut algorithm}
\begin{theorem}
\label{thm:formal-weighted}
For every $\eps\in(0,1)$, there is a randomized one-pass streaming algorithm that, given a weighted, undirected graph $G$ under general turnstile updates, uses $\Ot(n/\eps)$ bits of space and, with high probability, outputs a cut $\calC$ such that
$w_G(\calC)\le (1+\eps)\lambda(G).$ 
\end{theorem}

To prove this result, the main technical ingredient we require is a for-each cut sketch of the input graph. Given such a sketch, we can solve the $(1+\eps)$-approximate minimum cut problem as follows. By \Cref{thm:weighted-for-all-sparsifier}, we can store a for-all cut sparsifier with $\eps$ set to a constant using $\Ot(n)$ space in the general weighted turnstile model. Let $\widetilde\lambda$ denote the minimum-cut value of this sparsifier, and call a cut \emph{small} if its value in the sparsifier is at most a sufficiently large constant times $\widetilde\lambda$. Since the sparsifier approximates every cut within a constant factor, every small cut is a constant-factor approximate minimum cut of $G$. Karger's standard cut-counting bound~\cite{karger1993global} implies that there are only $n^{O(1)}$ such cuts. Moreover, a minimum cut of $G$ is itself small, so the small cuts form a polynomial-size candidate family containing a true minimum cut of $G$.

We can then query the for-each cut sketch with accuracy parameter $\eps$ and failure probability $1/\poly(n)$ on every candidate (small cut) and output the cut with the smallest estimated value. A union bound ensures that all candidate estimates are simultaneously correct. Since a minimum cut of $G$ belongs to the candidate family, the cut with the smallest estimated value cannot have value more than a $1+O(\eps)$ factor larger than the true minimum cut. This gives a $1+O(\eps)$ approximation, which can be made $1+\eps$ by rescaling. A complete proof of this reduction from for-each cut sparsification to approximate minimum cut appears in~\cite{ding2025}. Thus, the main technical ingredient we need is a for-each cut sketch in the general weighted turnstile model using $\Ot(n/\eps)$ space, and constructing such a sketch is our main technical result:

\begin{theorem}
\label{thm:changing-weights}
Fix $\eps\in(0,1)$ and $\delta\in(0,1/2)$. There is a randomized for-each cut sketch for weighted, undirected graphs $G$ in the general turnstile model such that, for every fixed query $\mathcal C\subseteq[n]$, the sketch returns an estimate $\widehat w(\mathcal C)$ satisfying
\[
(1-\eps)w_G(\mathcal C)
\leq \widehat w(\mathcal C)
\leq (1+\eps)w_G(\mathcal C)
\]
with probability at least $1-\delta$, using
$\Ot((n/\eps)\log(1/\delta))$ bits of space. The update time is $\Ot(\log(1/\delta))$. 
\end{theorem}

\begin{table}[t] \label{table:notation}
\centering
\scriptsize

\begin{minipage}[t]{0.48\linewidth}
\centering
\textbf{Streaming construction}\\[1mm]
\begin{tabular}{@{}p{0.34\linewidth}p{0.60\linewidth}@{}}
\toprule
Symbol & Description \\
\midrule
$b=4/\eps$
    & Number of buckets \\

    $h^L,h^R$
    & Two independent bucket hashes\\

$L_i,R_j$; $c=(L_i,R_j)$
    & Left/right buckets and their cell \\

\mbox{$\ct^L(u,R_j),\ \ct^R(L_i,v)$}
    & Exact aggregate-weight counters for a vertex--bucket pair. \\

\mbox{$\mathcal N^L(u,R_j), \mathcal N^R(L_i,v)$}
    & $L_1$-norm estimator blocks for a vertex--bucket pair. \\

\mbox{$\mathcal S^L(u,R_j), \mathcal S^R(L_i,v)$}
    & $\ell_1$-samplers for a vertex--bucket pair. \\
\bottomrule
\end{tabular}
\end{minipage}
\hfill
\begin{minipage}[t]{0.48\linewidth}
\centering
\textbf{Query and analysis}\\[1mm]
\begin{tabular}{@{}p{0.35\linewidth}p{0.59\linewidth}@{}}
\toprule
Symbol & Description \\
\midrule
$\calC$; $\chi_u$; $s_e$
    & Query cut, $\chi_u=\mathbf 1[u\in\calC]$, and cut-edge sign. \\

$\mpos{c},\mneg{c}$
    & Positive and negative cut-edge weights in cell $c$. \\

$\mtot{c},\mdif{c},\mmin{c}$
    & Total mass, signed difference, and minority mass in $c$. \\

$m^+,m^-$
    & Graph-wide positive and negative cut-edge weights. \\

$A_c$
    & $L_1$-norm estimate of $\mtot{c}$. \\

$F_c$
    & Indicator that the $L_1$-sample has the minority sign. \\

$\widehat m_c^{\mathrm{tot}}$; $\widehat w(\calC)$
    & Cell estimate and final cut estimate. \\
\bottomrule
\end{tabular}
\end{minipage}

\caption{Notation for the weighted for-each cut sketch.}
\label{tab:weighted-notation}
\end{table}

The \textbf{bipartite representation} of $G$ is the weighted
bipartite graph $H=(V^L,V^R,E_H)$, where $V^L$ and $V^R$ are copies of $V$ on the left and right sides.
For every edge $\{u,v\}\in E$ with $u<v$, the graph $H$
contains the single edge $(u^L,v^R)$ with weight $w_{uv}$. For a potential edge $\{u,v\} \notin E$, take $w_{uv}=0$. Thus, every edge of $G$ appears exactly
once in $H$.

We partition each side ($V^L,V^R$) into $b:=4/\eps$ \textbf{buckets}.
Fix independent pairwise-independent hash functions
$h^L,h^R:[n]\to[b]$. For $i,j\in[b]$, define the \textbf{left bucket}
$L_i=\{u^L:h^L(u)=i\}$ and the \textbf{right bucket}
$R_j=\{v^R:h^R(v)=j\}$. A \textbf{cell} is a pair
$c=(L_i,R_j)$. We identify this cell with its index pair $(i,j)$.

For every left vertex $u$ and right bucket $R_j$, define
\[
\ct^L(u,R_j)
=
\sum_{\substack{v>u\\ h^R(v)=j}} w_{uv},
\]
and, symmetrically, for every left bucket $L_i$ and right vertex $v$, define
\[
    \ct^R(L_i,v)
    =
    \sum_{\substack{u<v\\  h^L(u)=i}} w_{uv}.
\]

The sketch explicitly stores these two families of exact counters. The order of the arguments follows the left-to-right orientation of $H$.

$\mathcal{N}^L(u,R_j)$ denotes an $L_1$-norm estimator block for the edge-weight vector on edges $e=(u^L,v^R)$ of $H$ with $v^R\in R_j$, and $\mathcal{S}^L(u,R_j)$ denotes an $L_1$-sampler for the same edge-weight vector. Symmetrically, $\mathcal{N}^R(L_i,v)$ and $\mathcal{S}^R(L_i,v)$ denote the corresponding structures for the edge-weight vector on edges entering $v^R$ from $L_i$.

For each cell $c=(L_i,R_j)$, all left $L_1$-norm estimator blocks $\mathcal{N}^L(u,R_j)$ with $h^L(u)=i$ and all right blocks $\mathcal{N}^R(L_i,v)$ with $h^R(v)=j$ use the same linear map, so their sketches can be added and subtracted. The same holds for the corresponding $L_1$-samplers. Different cells use independent randomness, and the $L_1$-norm estimator blocks use randomness independent of the $L_1$-samplers. 

We now present our algorithm. See \Cref{fig:roadmap} for a summary of the steps and \Cref{tab:weighted-notation} for a table with the relevant notation.

\begin{figure}[t] 
\centering
\resizebox{1.1\linewidth}{!}{%
    \begin{tikzpicture}[
    >=Latex,
    font=\small,
    flow/.style={-{Latex[length=3mm,width=2mm]}, line width=1.2pt, color=colorFive},
    phase/.style={draw=none, rounded corners=6pt, line width=0.8pt, inner sep=8pt},
    stream/.style={draw=colorFive, rounded corners=5pt, line width=1pt, fill=colorTwo!15,
        align=center, minimum height=15mm, inner sep=6pt},
    alg/.style={draw=colorEight, rounded corners=5pt, line width=0.9pt, fill=colorFour!20,
        align=center, text width=39mm, minimum height=28mm, inner sep=5pt},
    algsel/.style={alg, fill=colorThree!15},
    algest/.style={alg, fill=colorOne!16},
    algcell/.style={alg, fill=colorSix!15},
    algout/.style={alg, fill=colorSeven!14},
    proof/.style={draw=colorNine, rounded corners=4pt, line width=0.8pt, fill=white,
        align=left, text width=39mm, minimum height=27mm, inner sep=5pt, font=\scriptsize},
    numbadge/.style={circle, draw=colorEight, fill=white, font=\bfseries\small,
        minimum size=6mm, inner sep=0pt}
]


\node[stream, text width=39mm] (updates) at (5.0,7.6)
    {\raggedright\textbf{Dynamic stream}\\updates to edge weights};

\node[stream, text width=96mm, right=15mm of updates] (state)
   {\raggedright\textbf{Maintain vertex--bucket structures}\\[1mm]
 1. exact counters: $\ct^L(u,R_j),\ \ct^R(L_i,v)$\\
 2. $L_1$-norm estimator blocks: $\mathcal N^L(u,R_j),\ \mathcal N^R(L_i,v)$\\
 3. $L_1$-samplers: $\mathcal S^L(u,R_j),\ \mathcal S^R(L_i,v)$\\
 an update to $(u^L,v^R)$ is applied to the structures for
$(u^L,R_{h^R(v)})$ and $(L_{h^L(u)},v^R)$\par};
\draw[flow] (updates) -- (state);

\begin{scope}[on background layer]
\node[phase, fit=(updates)(state)] (streamphase) {};
\end{scope}
\node[fill=white, inner sep=3.5pt, font=\bfseries\small, text=colorFive]
at (streamphase.north) {STREAMING MAINTENANCE};

\node[alg] (a1) at (1.9,2.7)
    {\textbf{Recover cell imbalance}\\[1mm]
 Receive query cut $\mathcal C$ and 
 for each cell $c$, combine the exact counters to recover
 \[
 \mdif{c}=m_c^+-m_c^-,
 \]
 which determines the minority sign in $c$.};

\node[algsel, right=7mm of a1] (a2)
    {{\textbf{Estimate $m_c^{\mathrm{tot}}$}\\[1mm]
 For each cell $c$, combine the $L_1$-norm sketches according to $\mathcal C$ and get approximation
 \[
 A_c\approx \mtot{c}=m_c^++m_c^-.
 \]}};

\node[algest, right=7mm of a2] (a3)
    {{\textbf{Estimate \\minority mass}\\[1mm]
 For each cell $c$, combine the $\ell_1$-samplers according to $\mathcal C$ and 
 let $F_c$ indicate whether the sampled edge has the minority sign, so
 \[
 \Pr[F_c=1]=\frac{m_c^{\min}}{m_c^{\mathrm{tot}}}.
 \]}};

\node[algcell, right=7mm of a3] (a4)
    {{\textbf{Estimate one cell}\\[1mm]
 Form 
 \[
 \widehat m_c^{\mathrm{tot}}
 =
 |\mdif{c}|+2A_cF_c.
 \]
 This is an unbiased estimate of
 $\mtot{c}=m_c^++m_c^-$.}};

\node[algout, right=7mm of a4] (a5)
    {{\textbf{Aggregate and amplify}\\[1mm]
 Sum the cell estimates
 \[
 \widehat w(\mathcal C)
 =
 \sum_{c\in[b]^2}\widehat m_c^{\mathrm{tot}}
 \]
 and return the \\median over independent repetitions.}};

\foreach \x/\n in {a1/1,a2/2,a3/3,a4/4,a5/5}{
    \node[numbadge] at ($ (\x.north west)+(2mm,-2mm) $) {\n};
}

\draw[flow] (a1) -- (a2);
\draw[flow] (a2) -- (a3);
\draw[flow] (a3) -- (a4);
\draw[flow] (a4) -- (a5);

\begin{scope}[on background layer]
\node[phase, fit=(a1)(a2)(a3)(a4)(a5)] (queryphase) {};
\end{scope}
\node[fill=white, inner sep=1.5pt, font=\bfseries\small, text=colorSix]
at (queryphase.north) {QUERY ALGORITHM};



\end{tikzpicture}%
}
\caption{Algorithm roadmap}
\label{fig:roadmap}
\end{figure}
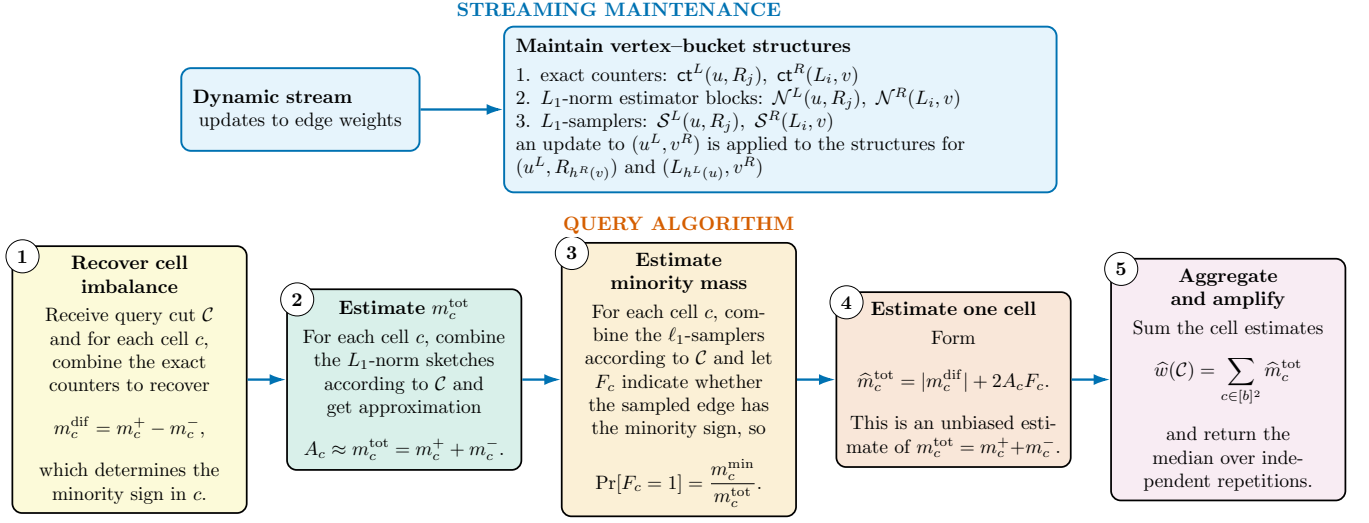

\begin{algorithm}[!t]
\caption{Weighted For-Each Cut Sketch}
\label{alg:weighted-for-each-sp}
\begin{algorithmic}[1]
    \item[] \textbf{Input:} A dynamic stream of updates
    $(\{u,v\},\Delta)$ and parameter
    $\eps\in(0,1)$.
    \item[] \textbf{Output:} Given a cut $\calC$, an estimate of
    its total weight $w_G(\calC)$.

    \item[] \textbf{Initialization:}
    \STATE Set $b:=4/\eps$, and fix $h^L,h^R:[n]\to[b]$.
    \STATE Maintain the exact counters
    $\ct^L(u,R_j)$ and $\ct^R(L_i,v)$.
    \STATE Maintain $\mathcal{N}^L(u,R_j)$,
    $\mathcal{N}^R(L_i,v)$, $\mathcal{S}^L(u,R_j)$, and
    $\mathcal{S}^R(L_i,v)$ defined above.

    \item[] \textbf{On an update of value $\Delta\in\mathbb Z$ to
    $\{u,v\}$:}
    \STATE Order the endpoints so that $u<v$, and view the update as acting
    on $e=(u^L,v^R)$.
    \STATE Add $\Delta$ to $\ct^L(u,R_{h^R(v)})$ and
    $\ct^R(L_{h^L(u)},v)$.
    \STATE Apply $\Delta$ to edge $e$ in
    $\mathcal{N}^L(u,R_{h^R(v)})$ and
    $\mathcal{N}^R(L_{h^L(u)},v)$, and in
    $\mathcal{S}^L(u,R_{h^R(v)})$ and
    $\mathcal{S}^R(L_{h^L(u)},v)$.

    \item[] \textbf{Recovery for a cut $\calC\subseteq[n]$:}
    \STATE Set $\chi_u=\mathbf{1}[u\in\calC]$ for every $u\in V$, and set
    $s_{uv}=\chi_u-\chi_v$ for every $u<v$.
    \STATE For every cell $c=(L_i,R_j)$, compute
    \[
    \mdif{c}=\sum_{u:h^L(u)=i}\chi_u\ct^L(u,R_j)
        -\sum_{v:h^R(v)=j}\chi_v\ct^R(L_i,v).
    \]
    \STATE For every cell $c=(L_i,R_j)$, form the temporary
    $L_1$-norm estimator block
    \[
        \sum_{u:h^L(u)=i}\chi_u\mathcal{N}^L(u,R_j)
        -\sum_{v:h^R(v)=j}\chi_v\mathcal{N}^R(L_i,v),
    \]
    and let $A_c$ be the unbiased $L_1$-norm estimate obtained by decoding it.
    \STATE For every cell $c=(L_i,R_j)$, form the temporary
    $\ell_1$-sampler
    \[
        \sum_{u:h^L(u)=i}\chi_u\mathcal{S}^L(u,R_j)
        -\sum_{v:h^R(v)=j}\chi_v\mathcal{S}^R(L_i,v),
    \]
    query it, and let $e_c$ be the edge returned, where a sampler on the
    zero vector returns $\bot$.
    \STATE If $\mdif{c}\geq0$, let
    $F_c=\mathbf{1}[e_c\neq\bot\text{ and }s_{e_c}=-1]$.
    Otherwise, let
    $F_c=\mathbf{1}[e_c\neq\bot\text{ and }s_{e_c}=1]$.
    \STATE Set $\widehat{m}_c^{\mathrm{tot}}=|\mdif{c}|+2A_cF_c.
    $
    \RETURN $\widehat w(\calC)=\sum_{c\in[b]^2}\widehat{m}_c^{\mathrm{tot}}.
    $
\end{algorithmic}
\end{algorithm}

Configure every $L_1$-sampler to have failure probability $n^{-C}$ for a
sufficiently large constant $C$. The number of samplers used in a single
repetition is $\poly(n)$, so by a union bound, with high probability, none
of them outputs $\mathtt{FAIL}$. The $L_1$-norm estimator blocks satisfy the
expectation and second-moment guarantees of
\Cref{lem:cauchy-l1-mass}. For a target failure probability $\delta$,
maintain $t:=O\bigl(\log(1/\delta)\bigr)$ independent repetitions of the
sketch and return the median of their estimates.

\begin{claim}
\label{clm:weighted-space-update}
This claim is deterministic and holds for every choice of $h^L,h^R$ and the randomness used by the $L_1$-norm estimator blocks and $L_1$-samplers. A single repetition of the sketch uses $\Ot(n/\eps)$ bits of space and has $\Ot(1)$ worst-case update time. The full sketch uses $\Ot((n/\eps)\log(1/\delta))$ bits of space and has $\Ot(\log(1/\delta))$ worst-case update time.
\end{claim}

\begin{proof}
One repetition stores two families of $nb$ exact counters, for a total of $2nb$ counters. Since all intermediate edge weights are bounded by $\poly(n)$, each counter has magnitude $\poly(n)$ and hence uses $O(\log n)$ bits. Thus, the counters use $\Ot(nb)=\Ot(n/\eps)$ bits of space. The descriptions of $h^L,h^R$ use $\poly \log(n,1/\eps)$ bits. The two families of $\mathcal{N}$ structures store one $L_1$-mass estimator block for each of the $nb$ vertex--bucket pairs, and the two families of $\mathcal{S}$ structures store one $L_1$-sampler for each of the $nb$ vertex--bucket pairs. Thus, one repetition stores $2nb$ $L_1$-mass estimator blocks and $2nb$ $L_1$-samplers. By \Cref{lem:cauchy-sketch}, each $L_1$-mass estimator block uses $\poly \log(n)$ bits of space, and each $L_1$-sampler uses $\poly \log(n)$ bits of space. Since $b=O(1/\eps)$, these structures use $\Ot(n/\eps)$ bits of space.

An update changes two exact counters. It also changes two $L_1$-mass estimator blocks and two $L_1$-samplers. Each $L_1$-mass estimator block and each $L_1$-sampler can be updated in $\poly \log(n)$ time, so the resulting update time is $\Ot(1)$ for one repetition. Multiplying by $t=O(\log(1/\delta))$ proves the claimed bounds.
\end{proof}

We now analyze just one repetition of the algorithm.

\subsection{Analysis}

Fix a query cut $\calC$ and analyze one repetition of the sketch. For every
vertex $u$, define
\[
    \chi_u=\mathbf{1}[u\in\calC].
\]
For every ordered pair $u<v$, define the \textbf{signed cut value}
\[
    s_{uv}
    =(\chi_u-\chi_v)
    \in\{-1,0,1\}.
\]
The value $s_{uv}$ is nonzero exactly when $\{u,v\}$ crosses the cut. Such an
edge is called \textbf{positive} when $s_{uv}=1$ i.e. the edge is leaving the
cut in $G$ and \textbf{negative} when $s_{uv}=-1$ i.e. the edge is entering
the cut in $G$. We also write $s_e=s_{uv}$ and $w_e=w_{uv}$ for the
corresponding edge $e=(u^L,v^R)$ of $H$.

For every cell $c$, define
\[
    \mpos{c}=\sum_{\substack{e\in c\\s_e=1}} w_e,
    \qquad
    \mneg{c}=\sum_{\substack{e\in c\\s_e=-1}} w_e,
\]
\[
    \mtot{c}=\mpos{c}+\mneg{c},
    \qquad
    \mdif{c}=\mpos{c}-\mneg{c},
    \qquad
    \mmin{c}=\min\{\mpos{c},\mneg{c}\}.
\]
Thus, $\mtot{c}$ is the total weight of cut edges in cell $c$, and
$\mmin{c}$ is the total weight of the sign having smaller total weight.
Define the total weights of positive and negative cut edges by
\[
    m^+=\sum_{e:s_e=1}w_e=\sum_{c\in[b]^2}\mpos{c},
    \qquad
    m^-=\sum_{e:s_e=-1}w_e=\sum_{c\in[b]^2}\mneg{c},
\]
and define $m^{\mathrm{tot}}=m^++m^-$. Since every edge of $G$ appears in
exactly one cell of $H$,
\[
    m^{\mathrm{tot}}=m^++m^-=w_G(\calC)
    =\sum_{c\in[b]^2}\mtot{c}.
\]
For every cell, we use the identity
\begin{equation}
    \label{eq:weighted-length-minority}
    \mtot{c}=|\mdif{c}|+2\mmin{c}.
\end{equation}

We first verify that the exact counters recover $\mdif{c}$.

\begin{claim}
    \label{clm:weighted-exact-cell-difference}
    For every fixed outcome of $h^L,h^R$ and every cell $c=(L_i,R_j)$,
    \[
        \mdif{c}
        =\sum_{u:h^L(u)=i}\chi_u\ct^L(u,R_j)
        -\sum_{v:h^R(v)=j}\chi_v\ct^R(L_i,v).
    \]
\end{claim}

\begin{proof}
    By the definitions of the two counters, the right-hand side equals
    \begin{align*}
        &\sum_{\substack{u<v\\h^L(u)=i,\ h^R(v)=j}}
        w_{uv}\chi_u
        -\sum_{\substack{u<v\\h^L(u)=i,\ h^R(v)=j}}
        w_{uv}\chi_v \\
        =&
\sum_{\substack{u<v\\h^L(u)=i,\ h^R(v)=j}}
        w_{uv}(\chi_u-\chi_v)
        =\sum_{e\in c}w_es_e
        =\mdif{c}.
    \end{align*}
\end{proof}

For every cell $c=(L_i,R_j)$, the combined sampler vector used by the recovery procedure has coefficient
$w_{uv}(\chi_u-\chi_v)=w_es_e$ 
on an edge $e=(u^L,v^R)$ in cell $c$. Consequently, its support consists exactly of the cut edges in $c$, the absolute value of the coordinate corresponding to an edge $e$ is $w_e$, and its sign is $s_e$. In particular, the $L_1$ norm of the combined vector is $\mtot{c}$.

We next record the distribution of the samples returned by the $\ell_1$-samplers.

\begin{lemma}
\label{lem:weighted-sample-recovery}
Fix any outcome of $h^L,h^R$ and condition on all queried $\ell_1$-samplers succeeding. For every cell $c$ with $\mtot{c}>0$ and every cut edge $e\in c$,
\[
\Pr[e_c=e]=\frac{w_e}{\mtot{c}}.
\]
If $\mtot{c}=0$, the combined sampler vector is zero.
\end{lemma}

\begin{proof}
As observed above, the combined sampler vector for cell $c$ has coordinate $w_es_e$ on each cut edge $e\in c$ and zero on every non-cut edge. Hence its $L_1$ norm is
$\sum_{\substack{e\in c\\s_e\neq0}}w_e=\mtot{c}.$
Therefore, by the guarantee of the $\ell_1$-sampler, the sampler returns a cut edge $e\in c$ with probability
$\frac{|w_es_e|}{\mtot{c}}=\frac{w_e}{\mtot{c}}.
$
\end{proof}

For a cell $c$ with $\mtot{c}>0$, let $F_c$ be the indicator that the sample $e_c$ has the minority sign, and define
\[
q_c=\frac{\mmin{c}}{\mtot{c}}.
\]
By \Cref{lem:weighted-sample-recovery}, the $L_1$-sample has the minority sign with probability $q_c$. Thus, $F_c\sim\operatorname{Bern}(q_c)$. If $\mtot{c}=0$, set $F_c=q_c=0$.
Let $A_c$ be the output of the $L_1$-mass estimator block for the combined vector of cell $c$. By \Cref{lem:cauchy-l1-mass},
\[
\E[A_c]=\mtot{c},
\qquad
\E[A_c^2]=C_r(\mtot{c})^2.
\]
Define the cell estimator as
\[
\widehat m_c^{\mathrm{tot}}
=|\mdif{c}|+2A_cF_c.
\]

\begin{claim}
\label{clm:weighted-cell-estimator}
Fix any outcome of $h^L,h^R$ and a cell $c$. We have
\[
\E[\widehat m_c^{\mathrm{tot}}]=\mtot{c},
\qquad
\Var(\widehat m_c^{\mathrm{tot}})
\leq 4C_r\mtot{c}\mmin{c}
=O(\mtot{c}\mmin{c}).
\]
\end{claim}
\begin{proof}
If $\mtot{c}=0$, the claim is immediate. Assume that
$\mdif{c}\geq0$. The other case is symmetric. Then
$\mpos{c}\geq\mneg{c}$, so the negative edges have the minority sign and
$\mmin{c}=\mneg{c}$. Thus,
\[
q_c=\frac{\mneg{c}}{\mtot{c}},
\qquad
F_c\sim\operatorname{Bern}(q_c),
\]
and
\[
\E[F_c]=q_c,
\qquad
\Var(F_c)=q_c(1-q_c).
\]

By \Cref{lem:cauchy-l1-mass},
$\E[A_c]=\mtot{c}$ and
$\Var(A_c)
=(C_r-1)(\mtot{c})^2.$
Moreover, $A_c$ is independent of $F_c$. Since
$|\mdif{c}|=\mpos{c}-\mneg{c}$,
\[
\begin{aligned}
\E[\widehat m_c^{\mathrm{tot}}]
&=\mpos{c}-\mneg{c}
+2\mtot{c}q_c \\
&=\mpos{c}+\mneg{c}
=\mtot{c}.
\end{aligned}
\]

For independent random variables $X$ and $Y$,
\[
\Var(XY)
=
\Var(X)\Var(Y)
+\Var(X)(\E[Y])^2
+\Var(Y)(\E[X])^2.
\]
Since $F_c^2=F_c$ and $A_c$ is independent of $F_c$,
\[
\begin{aligned}
\Var(A_cF_c)
&=\E[A_c^2]\E[F_c]-(\E[A_c]\E[F_c])^2 \\
&=C_r(\mtot{c})^2q_c-(\mtot{c})^2q_c^2 =(\mtot{c})^2q_c(C_r-q_c) \\
&\leq C_r(\mtot{c})^2q_c
=C_r\mtot{c}\mmin{c}.
\end{aligned}
\]

Finally, $|\mdif{c}|$ is fixed once $h^L,h^R$ are fixed, so
\[
\begin{aligned}
\Var(\widehat m_c^{\mathrm{tot}})
&=\Var\left(|\mdif{c}|+2A_cF_c\right)
=4\Var(A_cF_c) \\
&\leq 4C_r\mtot{c}\mmin{c}
=O(\mtot{c}\mmin{c}),
\end{aligned}
\]
where the last equality uses that $r$ is a fixed constant.
\end{proof}

For every cell,
\begin{equation}
    \label{eq:weighted-cell-minority-product}
    \mtot{c}\mmin{c}\leq2\mpos{c}\mneg{c}.
\end{equation}
Indeed, if $\mpos{c}\geq\mneg{c}$, then
$(\mpos{c}+\mneg{c})\mneg{c}\leq2\mpos{c}\mneg{c}$, and the other case is
symmetric. We now bound the sum of these products over all cells.

\begin{claim}
    \label{clm:weighted-opposite-collisions}
    For every fixed graph $G$ and cut $\calC$, the expectation over the
    independent bucket hashes $h^L,h^R$ satisfies
    \[
        \E_{h^L,h^R}\!\left[
            \sum_{c\in[b]^2}\mpos{c}\mneg{c}
        \right]
        =\frac{m^+m^-}{b^2}
        \leq\frac{w_G(\calC)^2}{4b^2}.
    \]
\end{claim}

\begin{proof}
    Consider a positive edge $e=(u^L,v^R)$ and a negative edge
    $e'=(u'^L,v'^R)$. We have
    $\chi_u=1$, $\chi_v=0$, $\chi_{u'}=0$, and $\chi_{v'}=1$. Hence
    $u\neq u'$ and $v\neq v'$. By pairwise independence of each hash
    function and independence between $h^L$ and $h^R$,
    \[
    \begin{aligned}
        \Pr[e\text{ and }e'\text{ are assigned to the same cell}]
        &=\Pr[h^L(u)=h^L(u')]\Pr[h^R(v)=h^R(v')] \\
        &=\frac{1}{b^2}.
    \end{aligned}
    \]
    The sum $\sum_c\mpos{c}\mneg{c}$ equals
    \[
        \sum_{\substack{e:s_e=1\\e':s_{e'}=-1}}
        w_e w_{e'}\,
        \mathbf{1}[e\text{ and }e'\text{ are assigned to the same cell}].
    \]
    Linearity of expectation therefore gives the equality. The inequality
    follows from
    \[
        m^+m^-\leq\frac{(m^++m^-)^2}{4}
        =\frac{w_G(\calC)^2}{4}.
    \]
\end{proof}

We can now prove the guarantee for one repetition. Let
$w=w_G(\calC)$. If $w=0$, every combined vector is zero and the algorithm
returns zero. Assume henceforth that $w>0$, and define
$\widehat w
    =\sum_{c\in[b]^2}\widehat m_c^{\mathrm{tot}}.$

Fix $h^L,h^R$. The sketch and sampler randomness is independent across
different cells, so the cell estimators are independent. By
\Cref{clm:weighted-cell-estimator} and its proof,
\[
    \E[\widehat w\mid h^L,h^R]=w
\]
and
\[
    \Var(\widehat w\mid h^L,h^R)
    \leq4C_r
    \sum_{c\in[b]^2}\mtot{c}\mmin{c}.
\]
So we have
\begin{align*}
    \E\left[(\widehat w-w)^2\right]
    &=
    \E_{h^L,h^R}\left[
        \Var(\widehat w\mid h^L,h^R)
    \right]
    && \tag{since $\E[\widehat w\mid h^L,h^R]=w$} \\
    &\leq
    4C_r
    \E_{h^L,h^R}\left[
        \sum_{c\in[b]^2}\mtot{c}\mmin{c}
    \right]
    && \tag{by \Cref{clm:weighted-cell-estimator}} \\
    &\leq
    8C_r
    \E_{h^L,h^R}\left[
        \sum_{c\in[b]^2}\mpos{c}\mneg{c}
    \right]
    && \tag{by \Cref{eq:weighted-cell-minority-product}} \\
    &\leq
    \frac{2C_rw^2}{b^2}\leq
    \frac{C_r\eps^2w^2}{8}
    && \tag{by \Cref{clm:weighted-opposite-collisions} and since $b\geq4/\eps$} 
\end{align*}
Chebyshev's inequality
therefore gives
\[
    \Pr[|\widehat w-w|>\eps w]
    \leq\frac{C_r}{8}
    <\frac{1}{6},
\]
where the final inequality uses $r=10$, for which $C_r<4/3$.

The analysis above treats the $\ell_1$-samplers as exact and works with the ideal Cauchy estimators. By choosing the failure probabilities in these primitives sufficiently small, the probability that any of these assumptions fails in one repetition is at most $n^{-10}$. Consequently, one repetition fails with probability at most
\[
    \frac{1}{6}+n^{-10}<\frac14.
\]
Thus, every repetition succeeds with probability at least $3/4$.

\begin{proof}[Proof of \Cref{thm:changing-weights}]
Run $t=O(\log(1/\delta))$ independent repetitions of
\Cref{alg:weighted-for-each-sp} and return the median estimate. The analysis
above shows that each repetition returns a $(1\pm\eps)$-approximation with
probability at least $3/4$, so \Cref{prop:chernoff} shows that the median
fails with probability at most $\delta$. The space bound follows from
\Cref{clm:weighted-space-update}.
\end{proof}

\section{\texorpdfstring{$\ell_p$}{Lp}-Norm \texorpdfstring{$k$}{k}-Partitioning}
The proof follows the same general strategy as our reduction from
for-each cut sparsification to minimum cut. We first use an independent
constant-accuracy for-all sparsifier to identify a polynomial-size family
of candidate $k$-partitions containing an optimum solution. For a
$k$-partition $\mathcal P=(S_1,\ldots,S_k)$, define
\[
\Phi_{p,G}(\mathcal P)
:=
\left(\sum_{i=1}^k w_G(S_i)^p\right)^{1/p}
\]
for $p<\infty$, and
\[
\Phi_{\infty,G}(\mathcal P)
:=
\max_{i\in[k]} w_G(S_i).
\]
Thus, the objective depends only on the values of the $k$ ordinary cuts
$S_1,\ldots,S_k$, so both the for-all sparsifier and our for-each
sketch can be used to estimate it through ordinary cut queries.
The key observation is that an optimum $\ell_p$-norm $k$-partition is
a constant-factor approximate minimum $k$-cut. The counting argument of
Chekuri, Quanrud, and Xu~\cite{chekuri2020lp} therefore implies that,
for fixed $k$, there are only polynomially many candidates. We can query
our for-each sketch on the $k$ parts of every candidate with
sufficiently small failure probability and take a union bound over all
queries. Selecting the candidate with minimum estimated $\ell_p$-norm
then gives a $(1+\eps)$-approximate minimum $\ell_p$-norm
$k$-partition, while the total space remains $\Ot(n/\eps)$.

\begin{corollary} \label{cor:kpartition}
    Fix a constant integer $k\geq2$ and let $p\in[1,\infty]$.
For every $\eps\in(0,1)$, there is a randomized one-pass streaming
algorithm that, given a weighted, undirected graph $G$ under general
turnstile updates whose final edge weights are nonnegative, uses
$\Ot(n/\eps)$ bits of space and, with high probability, outputs a
partition $\mathcal P$ of $V$ into $k$ nonempty parts satisfying
\[
\Phi_{p,G}(\mathcal P)
\leq
(1+\eps)\lambda_{k,p}(G),
\]
where
\[
\lambda_{k,p}(G)
:=
\min_{\mathcal P}\Phi_{p,G}(\mathcal P).
\] The update time is $t+ \poly \log n$ where $t$ is the update time of a constant-approximation for-all cut sparsifier.
\end{corollary}

\begin{proof}
Set $\gamma:=\eps/3$. Independently maintain a $(1\pm1/4)$ for-all cut
sparsifier $H$ using \Cref{thm:weighted-for-all-sparsifier} and a
for-each cut sketch for $G$ using \Cref{thm:changing-weights} with
accuracy parameter $\gamma$ and a sufficiently small inverse-polynomial
failure probability.

Let $\lambda_k(G)$ denote the ordinary minimum $k$-cut value, and let
$\mathcal P^\star=(S_1^\star,\ldots,S_k^\star)$ be a minimum
$\ell_p$-norm $k$-partition of $G$. We first show that
$\mathcal P^\star$ is a constant-factor approximate minimum $k$-cut.
For every $k$-partition $\mathcal P=(S_1,\ldots,S_k)$,
\begin{equation}
\label{eq:k-cut-as-cuts-lp}
w_G(\mathcal P)
=
\frac12\sum_{i=1}^k w_G(S_i),
\end{equation}
since every edge crossing the partition is counted once from each of
its two endpoint parts.

Let $\mathcal P_\Sigma$ be a minimum $k$-cut of $G$. For $p<\infty$,
H\"older's inequality gives
\[
w_G(\mathcal P^\star)
=
\frac12\sum_{i=1}^k w_G(S_i^\star)
\leq
\frac12 k^{1-1/p}\Phi_{p,G}(\mathcal P^\star).
\]
Since $\mathcal P^\star$ minimizes $\Phi_{p,G}$,
\[
\Phi_{p,G}(\mathcal P^\star)
\leq
\Phi_{p,G}(\mathcal P_\Sigma).
\]
Moreover, every part $S_i$ of $\mathcal P_\Sigma$ has
$w_G(S_i)\leq\lambda_k(G)$, since every edge in its boundary crosses
the partition. Hence
\[
\Phi_{p,G}(\mathcal P_\Sigma)
\leq
k^{1/p}\lambda_k(G),
\]
and therefore
\[
w_G(\mathcal P^\star)
\leq
\frac{k}{2}\lambda_k(G).
\]
For $p=\infty$, the same conclusion follows from
\[
w_G(\mathcal P^\star)
\leq
\frac{k}{2}\Phi_{\infty,G}(\mathcal P^\star)
\leq
\frac{k}{2}\Phi_{\infty,G}(\mathcal P_\Sigma)
\leq
\frac{k}{2}\lambda_k(G).
\]
Thus, for every $p\in[1,\infty]$, an optimum $\ell_p$-norm
$k$-partition is a $(k/2)$-approximate minimum $k$-cut.

On the event that $H$ preserves every cut, it also preserves the value
of every $k$-cut:
\[
\frac34 w_G(\mathcal P)
\leq
w_H(\mathcal P)
\leq
\frac54 w_G(\mathcal P)
\]
simultaneously for every $k$-partition $\mathcal P$. In particular,
\[
\frac34\lambda_k(G)
\leq
\lambda_k(H)
\leq
\frac54\lambda_k(G).
\]

If $\lambda_k(H)=0$, output any $k$-partition $\mathcal P$ attaining
this value. The preceding inequalities imply that $w_G(\mathcal P)=0$,
and hence $\Phi_{p,G}(\mathcal P)=\lambda_{k,p}(G)=0$. Assume
henceforth that $\lambda_k(H)>0$. Define
\[
\mathcal F_k
:=
\left\{
\mathcal P:
w_H(\mathcal P)\leq k\lambda_k(H)
\right\}.
\]
Since
\[
w_H(\mathcal P^\star)
\leq
\frac54 w_G(\mathcal P^\star)
\leq
\frac{5k}{8}\lambda_k(G)
\leq
\frac{5k}{6}\lambda_k(H)
<
k\lambda_k(H),
\]
we have $\mathcal P^\star\in\mathcal F_k$. Conversely, every
$\mathcal P\in\mathcal F_k$ satisfies
\[
w_G(\mathcal P)
\leq
\frac43w_H(\mathcal P)
\leq
\frac{4k}{3}\lambda_k(H)
\leq
\frac{5k}{3}\lambda_k(G).
\]
Thus, every partition in $\mathcal F_k$ is a $(5k/3)$-approximate
minimum $k$-cut of $G$. The counting argument of Chekuri, Quanrud, and
Xu~\cite{chekuri2020lp} implies that, for fixed $k$ and
$\alpha\geq1$, the number of $\alpha$-approximate minimum $k$-cut
partitions is $n^{O(\alpha k)}$. Since $5k/3$ is a constant for fixed
$k$, it follows that
$|\mathcal F_k|
\leq
n^{D_k}$
for some constant $D_k$ depending only on $k$.

For each $\mathcal P=(S_1,\ldots,S_k)\in\mathcal F_k$, query the
for-each sketch on $S_1,\ldots,S_k$. For $p<\infty$, define
\[
\widehat\Phi_p(\mathcal P)
:=
\left(
\sum_{i=1}^k \widehat w(S_i)^p
\right)^{1/p},
\]
and for $p=\infty$ define
$\widehat\Phi_\infty(\mathcal P)
:=
\max_{i\in[k]}\widehat w(S_i).$
The for-all cut sparsifier $H$ is independent of the for-each sketch,
so conditioned on $H$, the family of queried cuts is fixed independently
of the for-each sketch's randomness. Configure each cut query to fail
with probability at most $n^{-(D_k+C+2)}$, where $C>0$ is a sufficiently
large constant. Since there are at most $kn^{D_k}$ queries, a union bound
implies that, with high probability,
\[
(1-\gamma)w_G(S_i)
\leq
\widehat w(S_i)
\leq
(1+\gamma)w_G(S_i)
\]
simultaneously for every part of every
$\mathcal P\in\mathcal F_k$. Consequently,
\[
(1-\gamma)\Phi_{p,G}(\mathcal P)
\leq
\widehat\Phi_p(\mathcal P)
\leq
(1+\gamma)\Phi_{p,G}(\mathcal P)
\]
simultaneously for every $\mathcal P\in\mathcal F_k$ and every
$p\in[1,\infty]$.

Let $\widehat{\mathcal P}$ minimize
$\widehat\Phi_p(\mathcal P)$ over $\mathcal F_k$. Since
$\mathcal P^\star\in\mathcal F_k$,
\[
\Phi_{p,G}(\widehat{\mathcal P})
\leq
\frac{\widehat\Phi_p(\widehat{\mathcal P})}{1-\gamma}
\leq
\frac{\widehat\Phi_p(\mathcal P^\star)}{1-\gamma}
\leq
\frac{1+\gamma}{1-\gamma}\lambda_{k,p}(G)
\leq
(1+\eps)\lambda_{k,p}(G),
\]
where the last inequality follows from $\gamma=\eps/3$ and
$\eps\in(0,1)$.

The constant-accuracy for-all sparsifier uses $\Ot(n)$ bits of space.
Moreover,
\[
\log\!\left(n^{D_k+C+2}\right)=O_k(\log n),
\]
so \Cref{thm:changing-weights} uses $\Ot(n/\eps)$ bits of space for the
required per-query failure probability.

Since $k$ is constant, the failure probability
$\delta=n^{-(D_k+C+2)}$ satisfies
$\log(1/\delta)=O_k(\log n)$. Thus, by
\Cref{thm:changing-weights}, the for-each cut sketch
has $\poly\log n$ update time. Since each update is
processed by both sketches, the total update time is
$t+\poly\log n$, where $t$ is the update time of the
constant-accuracy for-all cut sparsifier.
\end{proof}

\begin{remark}
The case $p=1$ recovers minimum $k$-cut: for every partition
$\mathcal P$,
\[
\Phi_{1,G}(\mathcal P)
=
\sum_{i=1}^k w_G(S_i)
=
2w_G(\mathcal P).
\]
Thus, minimizing the $\ell_1$ objective gives exactly the same partition
as minimum $k$-cut, and dividing the objective value by two gives the
ordinary $k$-cut value.

The same computation also yields an estimate of the optimum
$\ell_p$-norm value. On the high-probability event above,
\[
\widehat\lambda_{k,p}
:=
\min_{\mathcal P\in\mathcal F_k}
\widehat\Phi_p(\mathcal P)
\]
satisfies
\[
(1-\gamma)\lambda_{k,p}(G)
\leq
\widehat\lambda_{k,p}
\leq
(1+\gamma)\lambda_{k,p}(G).
\]
This also shows that our space bound for estimating the optimum value is
optimal up to polylogarithmic factors. Indeed, the
$\Omega(n/\eps)$ lower bound for estimating ordinary minimum
cut~\cite{ding2025} transfers to minimum $\ell_p$-norm
$k$-partitioning for every fixed $k$ and $p$. Given a graph $G$, adjoin
$k-2$ isolated vertices. Then, for $p<\infty$,
\[
\lambda_{k,p}\!\left(G\cup I_{k-2}\right)
=
2^{1/p}\lambda_2(G),
\]
while
\[
\lambda_{k,\infty}\!\left(G\cup I_{k-2}\right)
=
\lambda_2(G).
\]
Thus, a $o(n/\eps)$-space algorithm for estimating the optimum
$\ell_p$-norm $k$-partitioning value would imply a
$o(n/\eps)$-space algorithm for estimating ordinary minimum cut.
\end{remark}

\section{Two-cluster correlation clustering}

We now apply our for-each cut sketch to evaluating weighted two-cluster correlation clusterings. For every pair $\{u,v\}$, let $w^+_{uv}$ and $w^-_{uv}$ denote the nonnegative costs of placing $u$ and $v$ in different clusters and in the same cluster, respectively. A two-clustering can be represented by labels $x_v\in\{0,1\}$ for each vertex $v$, with disagreement cost
\[\operatorname{dis}_G(x)
:=
\sum_{u<v}
\left(
w^+_{uv}\mathbf{1}[x_u\neq x_v]
+
w^-_{uv}\mathbf{1}[x_u=x_v]
\right).\]
We impose no probability constraints on the positive and negative
weights. The key observation is that this disagreement cost can be expressed as the weight of a single cut in an auxiliary graph on two copies of the vertex set. Applying our for-each cut sketch to this graph then gives the desired space bound under general turnstile updates.

\begin{corollary}
\label{cor:correlation-clustering}
For every $\eps\in(0,1)$, there is a randomized one-pass streaming
algorithm that, given general turnstile updates to the positive and
negative constraint weights of a weighted correlation-clustering
instance on $n$ vertices whose final weights are nonnegative,
uses $\Ot(n/\eps)$ bits of space and, with high probability,
estimates the disagreement cost of any fixed two-clustering
within a
$(1\pm\eps)$ factor. The update time is $\poly\log n$.
\end{corollary}

\begin{proof}
Construct an auxiliary graph $H$ on two copies $v^0,v^1$ of every
vertex $v\in V$. Fix an arbitrary ordering of the vertices. For every
pair $\{u,v\}$ with $u<v$, place an edge $\{u^0,v^0\}$ of weight
$w^+_{uv}$ for its positive constraint and an edge $\{u^0,v^1\}$ of
weight $w^-_{uv}$ for its negative constraint.

Each update to $w^+_{uv}$ or $w^-_{uv}$ corresponds to an update
of the same magnitude to the associated edge of $H$. Thus, we
can maintain a for-each cut sketch of $H$ directly under general
turnstile updates, without explicitly storing $H$.

Consider a two-clustering and let $x_v\in\{0,1\}$ indicate the cluster
containing $v$. Define the cut
\[
\calC_x
:=
\{v^0:x_v=1\}
\cup
\{v^1:x_v=0\}.
\]
The edge $\{u^0,v^0\}$ crosses $\calC_x$ exactly when $x_u\neq x_v$, while
the edge $\{u^0,v^1\}$ crosses $\calC_x$ exactly when $x_u=x_v$. Hence
\[
w_H(\calC_x)
=
\sum_{u<v} w^+_{uv}\mathbf 1[x_u\neq x_v]
+
\sum_{u<v} w^-_{uv}\mathbf 1[x_u=x_v],
\]
which is exactly the disagreement cost of the clustering.

Thus, evaluating a fixed two-clustering reduces to a single cut
query in a graph on $2n$ vertices. Since the clustering is fixed
independently of the sketch randomness, so is the corresponding
cut $\calC_x$. Applying \Cref{thm:changing-weights} to $H$ gives a
$(1\pm\eps)$ estimate using $\Ot(n/\eps)$ bits of space and
polylogarithmic update time.
\end{proof}

\begin{remark}
Our for-each cut sketch also applies to evaluating clusterings with
arbitrarily many clusters when the negative constraints are sparse.
Suppose $G^-$ has at most $s$ nonzero edges. Maintain an
$\Ot(n/\eps)$-space for-each cut sketch of $G^+$ and an
$\Ot(s)$-space turnstile sparse-recovery sketch of $G^-$.
For any fixed clustering $\calP=(S_1,\ldots,S_k)$, its disagreement
cost is
\[
\frac12\sum_{i=1}^k w_{G^+}(S_i)
+
\sum_{i=1}^k\sum_{\{u,v\}\subseteq S_i}w^-_{uv}.
\]
The second term is recovered exactly, while the first is
estimated using at most $n$ cut queries. A union bound and
nonnegativity of both terms give a $(1\pm\eps)$-approximation
using $\Ot(n/\eps+s)$ bits of space under general turnstile
updates. This improves the $\eps$-dependence of the
$\Ot(n/\eps^2+s)$-space baseline obtained by combining a
for-all cut sparsifier with exact storage of negative edges,
as in \cite{ACGMW2021correlation}.
\end{remark}

\section{Balanced directed graphs}

We now extend our for-each cut sketch to balanced directed graphs.
For $A,B\subseteq V$, let
$w_G(A,B)=\sum_{u\in A,v\in B}w_{uv}$, where $w_{uv}$ denotes the
weight of the directed edge $(u,v)$, taken to be zero if the edge
is absent. A directed graph is $\beta$-balanced if
$w_G(\calC,V\setminus\calC)\leq\beta\,w_G(V\setminus\calC,\calC)$
for every nontrivial cut $\calC\subsetneq V$. Unlike in the
undirected setting, our goal is to estimate the weight of edges
leaving the queried cut, rather than the total weight of edges
crossing it. We adapt the construction of our for-each cut sketch (\Cref{thm:changing-weights}) to distinguish
outgoing from incoming edges and exploit the balance condition
to obtain an $\Ot(n\sqrt{\beta}/\eps)$-space for-each cut sketch
under general turnstile updates.

\begin{theorem}
\label{thm:balanced-directed}
For every $\beta\geq 1$ and $\eps\in(0,1)$, there is a randomized
one-pass streaming algorithm that, given a weighted directed graph
under general turnstile updates (under the bounded-integer stream
assumptions of Section~2.1) whose final edge weights are
nonnegative and whose final graph is $\beta$-balanced, maintains a
for-each cut sketch using $\Ot(n\sqrt\beta/\eps)$ bits of space and
$\poly\log n$ worst-case update time. For every fixed nontrivial
$\calC\subsetneq V$ chosen independently of the sketch randomness, the
sketch returns an estimate $\widehat w_G(\calC)$ such that
\[
(1-\eps)w_G(\calC,V\setminus \calC)
\leq \widehat w_G(\calC)
\leq (1+\eps)w_G(\calC,V\setminus \calC)
\]
with high probability.
\end{theorem}

\begin{proof}
We adapt the construction of \Cref{thm:changing-weights}. Set $b=2^{\lceil\log_2(8\sqrt{\beta}/\eps)\rceil}.$
If $b>n$, we store the entire directed graph explicitly using
$\Ot(n^2)=\Ot(n\sqrt{\beta}/\eps)$ bits. Henceforth assume $b\leq n$.

As for \Cref{thm:changing-weights}, create left and right copies of the vertex set,
but now represent each directed edge $(u,v)$ by $(u^L,v^R)$,
without reordering its endpoints. We remove the restriction $u<v$ from the counters and sketches
of Section~3, treating each ordered pair $(u,v)$ as a distinct
coordinate. We maintain the same
vertex--bucket counters, $\ell_1$-norm estimator blocks, and
$\ell_1$-samplers, using independent bucket hashes
$h^L,h^R:[n]\to[b]$.

Fix a query cut $\calC$ and set $\chi_u=\mathbf{1}[u\in\calC]$.
For each directed edge $e=(u,v)$, define
$s_e=\chi_u-\chi_v$. Thus, $s_e=1$ for edges leaving $\calC$,
$s_e=-1$ for edges entering $\calC$, and $s_e=0$ otherwise.
For each cell $c$, let $m_c^+$ and $m_c^-$ be the total
weights of edges with signs $1$ and $-1$, respectively, and define
\[
m_c^{\mathrm{tot}}=m_c^++m_c^-,
\qquad
m_c^{\mathrm{dif}}=m_c^+-m_c^-,
\qquad
m_c^{\min}=\min\{m_c^+,m_c^-\}.
\]
Exactly as in \Cref{thm:changing-weights}, the counters recover
$m_c^{\mathrm{dif}}$ exactly, while the combined sketches
estimate $m_c^{\mathrm{tot}}$.

Let $\widehat m_c^{\mathrm{tot}}$ be the ideal cell estimator
from Section~3. We estimate the outgoing weight in cell $c$ by
\[
\widehat m_c^+
:=
\frac{\widehat m_c^{\mathrm{tot}}
      +m_c^{\mathrm{dif}}}{2}.
\]
Since
$m_c^+=(m_c^{\mathrm{tot}}+m_c^{\mathrm{dif}})/2$,
\Cref{clm:weighted-cell-estimator} and
\eqref{eq:weighted-cell-minority-product} give
\[
\mathbb E[\widehat m_c^+]=m_c^+,
\qquad
\operatorname{Var}(\widehat m_c^+)
\leq C_r m_c^{\mathrm{tot}}m_c^{\min}
\leq 2C_r m_c^+m_c^-,
\]
where $C_r<4/3$. These statements are
conditional on the bucket hashes.

Write
\[
m^+=w_G(\calC,V\setminus\calC),
\qquad
m^-=w_G(V\setminus\calC,\calC),
\]
and let $\widehat w=\sum_c\widehat m_c^+$.
The cell estimators are conditionally independent, and
$\mathbb E[\widehat w\mid h^L,h^R]=m^+$.
Therefore,
\[
\operatorname{Var}(\widehat w)
\leq
2C_r\,\mathbb E\left[\sum_c m_c^+m_c^-\right].
\]

The collision argument of \Cref{clm:weighted-opposite-collisions}
applies unchanged: an outgoing edge and an incoming edge have
distinct left endpoints and distinct right endpoints, so they
occupy the same cell with probability $1/b^2$. Consequently,
\[
\mathbb E\left[\sum_c m_c^+m_c^-\right]
=\frac{m^+m^-}{b^2}.
\]
By $\beta$-balance, applied to the complementary cut,
$m^-\leq\beta m^+$. Hence
\[
\operatorname{Var}(\widehat w)
\leq
\frac{2C_r\beta(m^+)^2}{b^2}.
\]

If $m^+=0$, balance also implies $m^-=0$, and the estimate is
zero. Otherwise, Chebyshev's inequality (\Cref{prop:chebyshev}) and the choice of $b$
give
\[
\Pr\left[
|\widehat w-m^+|>\frac{\eps m^+}{2}
\right]
\leq
\frac{8C_r\beta}{b^2\eps^2}
\leq\frac{C_r}{8}<\frac16.
\]

Finally, we replace the ideal primitives by their finite-space
implementations, as in \Cref{thm:changing-weights}. Choose the primitive error
parameters sufficiently small that, except with
inverse-polynomial probability, all sampler outputs agree
with their ideal counterparts and every norm estimate has
additive error at most
$\eta m_c^{\mathrm{tot}}$, where
$\eta\leq\eps/(2(1+\beta))$.
This is possible with polylogarithmic space per primitive,
since $b\leq n$ implies $\beta$ and $1/\eps$ are
polynomially bounded in $n$. The resulting total error is at most
\[
\eta\sum_c m_c^{\mathrm{tot}}
=\eta(m^++m^-)
\leq\frac{\eps m^+}{2}.
\]
Thus, one repetition succeeds with constant probability, and
taking the median of $O(\log n)$ independent repetitions gives
high probability.

Each repetition maintains $O(nb)$ vertex--bucket structures,
and each update modifies only a constant number of them.
The total space is
$\Ot(nb)=\Ot(n\sqrt{\beta}/\eps)$ bits and the update time
is $\poly\log n$.
\end{proof}

\section*{AI disclosure}
ChatGPT 5.6 Sol was used in the preparation of this manuscript. It was used to a limited extent for rewording, presentation, and figure design.
AI assistance also contributed to the development of part of the technical result. In particular, while discussing several candidate hard instances and attempting to find counterexamples to them, the authors, in interaction with Sol, identified the bucketing idea underlying the algorithm. The subsequent development of the algorithm and its analysis was carried out by the authors. In addition, the authors brainstormed with various Gemini models. Part of the results of this brainstorming was the identification of the $\ell_p$-norm $k$-partitioning corollary. The other two extensions were identified and written by Sol. The authors assume responsibility for all content.

\bibliographystyle{alpha}
\bibliography{ref}

\appendix 

\section{Omitted proofs from preliminaries}\label{appen:prelim}
We prove \Cref{lem:cauchy-sketch} using the following two lemmas. 
\begin{lemma} \label{appen:lemCV}
Fix a constant integer $r>2$, let $x\in\mathbb{Z}^d$, and let $\eta\in(0,1/2)$. There is a finite-precision linear sketch $\widetilde X_i = \sum_{j=1}^d \widetilde\rho_{i,j}x_j$ for $i\in[r]$, such that, for \[ \widetilde A_r(x) = \cos^r\!\left(\frac{\pi}{2r}\right) \prod_{i=1}^r|\widetilde X_i|^{1/r}, \] we have $\Pr\left[ |\widetilde A_r(x)-A_r(x)| > \eta\|x\|_1 \right] \le \eta. $ Each $\widetilde\rho_{i,j}$ can be represented using $O(\log(d/\eta))$ bits. Moreover, each $\widetilde X_i$ can be represented using $O\!\left(\log(1+\|x\|_1)+\log(d/\eta)\right)$ bits.
\end{lemma}
\begin{proof}
If $x=0$, then $\widetilde A_r(x)=A_r(x)=0$, so the claim is
immediate. Henceforth assume $x\neq 0$.

Let
\[
    T=\frac{4r}{\eta}
    \qquad\text{and}\qquad
    \varepsilon=\frac{(\eta/r)^r}{(T+1)^{r-1}}.
\]
By the bounded-precision construction of \cite{Indyk2006}, for any $\delta>0$ we can construct $\widetilde\rho_{i,j}$ for each $\rho_{i,j}$ such that 
$\Pr\left[|\widetilde\rho_{i,j}-\rho_{i,j}|>\varepsilon\right]\le\delta,$
using $O(\log(1/\delta+1/\varepsilon))$ bits. Taking $\delta=\eta/(2rd)$ and applying a union bound, with probability at least $1-\eta/2$,
$|\widetilde\rho_{i,j}-\rho_{i,j}|\le\varepsilon$
simultaneously for all $i\in[r]$ and $j\in[d]$. On this event,
\[
    |\widetilde X_i-X_i|\le\varepsilon\|x\|_1
    \qquad\text{for every }i\in[r].
\]

By Cauchy $1$-stability, $X_i/\|x\|_1$ is a standard Cauchy random variable. Since a standard Cauchy variable $Z$ satisfies
$\Pr[|Z|>T]\le\frac{2}{\pi T},$
our choice of $T$ and a union bound imply that, with probability at least $1-\eta/2$,
\[
    |X_i|\le T\|x\|_1
    \qquad\text{for every }i\in[r].
\]
Hence, with probability at least $1-\eta$, both events hold simultaneously, and therefore $|\widetilde X_i|\le(T+1)\|x\|_1$
    for every $i\in[r].$
Set $q=1/r$. Using
$\bigl||u|^q-|v|^q\bigr|\le|u-v|^q
$
and telescoping the two products, we obtain
\[
\begin{aligned}
    |\widetilde A_r(x)-A_r(x)|
    &\le r\varepsilon^{1/r}(T+1)^{(r-1)/r}\|x\|_1 \\
    &=\eta\|x\|_1.
\end{aligned}
\]
This proves the claimed coupling guarantee.

Finally, $\log(1/\delta+1/\varepsilon)=O(\log(d/\eta))$, so \cite{Indyk2006} gives the claimed coefficient precision. Thus each maintained measurement $\widetilde X_i$ can be represented using $O\!\left(\log(1+\|x\|_1)+\log(d/\eta)\right)$ bits.
\end{proof}

\begin{lemma}
\label{lem:cauchy-pseudorandomness}
Consider the finite-precision Cauchy sketch of
\Cref{appen:lemCV}. The independent coefficient blocks
$\bigl(\widetilde\rho_{1,j},\ldots,\widetilde\rho_{r,j}\bigr)$ for $j\in[d]$
can be replaced by pseudorandom coefficient blocks determined by a seed
of $O(\log^2(d/\eta))$ bits. For every fixed stream, this replacement changes the
probability of any event determined by the final sketch state by at most
$\eta$. Moreover, the coefficient block associated with any coordinate
$j$ can be regenerated from the seed in $\poly \log(d/\eta)$ time, and
hence each update can be processed in $\poly \log(d/\eta)$ time.
\end{lemma}

\begin{proof}
Group the random bits used to generate
$\widetilde\rho_{1,j},\ldots,\widetilde\rho_{r,j}$ into a single random
block $R_j$. By \Cref{appen:lemCV}, each such block
has length $O(\log(d/\eta))$, and the $r$ measurements can be stored
using a polylogarithmic number of bits.

Applying the randomness-reduction argument of
\cite[Lemma~3]{Indyk2006} replaces the independent blocks
$R_1,\ldots,R_d$ by pseudorandom blocks generated from a seed of
$O(\log^2(d/\eta))$ bits, while changing the probability of any event
determined by the final sketch state by at most $\eta$. The same block
$R_j$ can be regenerated whenever coordinate $j$ is updated. By
\cite{Indyk2006}, this takes $\poly \log(d/\eta)$ time per update.
\end{proof}

\begin{proof}[Proof of \Cref{lem:cauchy-sketch}]
Fix any constant $C>0$ and set $\eta=d^{-C}/4$.
By \Cref{appen:lemCV}, the finite-precision estimator can be
coupled with the ideal estimator with additive error at most
$\eta\|x\|_1$, except with probability $\eta$.
By \Cref{lem:cauchy-pseudorandomness}, the pseudorandom and
fully independent finite-precision sketch states can be coupled
to agree except with probability $\eta$. Therefore,
\[
\Pr\left[
|\widetilde A_r(x)-A_r(x)|>d^{-C}\|x\|_1
\right]
\le 2\eta\le d^{-C}.
\]
The space and update-time bounds follow from the two lemmas
and the polynomial bounds on the stream.
\end{proof}

\end{document}